\documentclass[acmsmall]{acmart}

\AtBeginDocument{%
  \providecommand\BibTeX{{%
    \normalfont B\kern-0.5em{\scshape i\kern-0.25em b}\kern-0.8em\TeX}}}

\setcopyright{rightsretained}
\copyrightyear{20xx}
\acmYear{20yy}

\usepackage{hyperref}
\usepackage{xfrac}
\usepackage{prettyref}

\newcommand{\rref}[2][]{\prettyref{#2}}
\newrefformat{sec}{Section\,\ref{#1}}
\newrefformat{app}{Appendix\,\ref{#1}}
\newrefformat{def}{Def.\,\ref{#1}}
\newrefformat{thm}{Theorem\,\ref{#1}}
\newrefformat{prop}{Proposition\,\ref{#1}}
\newrefformat{lem}{Lemma\,\ref{#1}}
\newrefformat{cor}{Corollary\,\ref{#1}}
\newrefformat{rem}{Remark\,\ref{#1}}
\newrefformat{ex}{Example\,\ref{#1}}
\newrefformat{tab}{Table\,\ref{#1}}
\newrefformat{fig}{Fig.\,\ref{#1}}
\newrefformat{case}{case\,\ref{#1}}
\newrefformat{foot}{Footnote\,\ref{#1}}

\usepackage{amsthm}

\usepackage{amsmath}
\usepackage{amssymb}
\usepackage{natbib}
\usepackage{mathtools}
\usepackage{tensor}
\usepackage{bm}
\usepackage{tikz}
\usepackage{tikz-cd}
\usepackage{graphicx}
\usepackage{stmaryrd}

\definecolor{semblue}{rgb}{0,0,0.7}
\definecolor{vgreen}{rgb}{.1,.5,0}%
\definecolor{vdarkgreen}{rgb}{.06,.3,0}%
\definecolor{vred}{rgb}{.7,0,0}%
\definecolor{vblue}{rgb}{.1,.15,.62}%
\definecolor{vgray}{rgb}{.35,.35,.35}
\definecolor{darkishgray}{rgb}{.35,.35,.35}
\definecolor{vvblue}{rgb}{.14,.21,.868}%

\usepackage{xcolor}
\usepackage{float}
\usepackage[prefixflatinterpret,prefixflatbracketmodalinterpret,prefixflatsetinterpret,sidenotecalculus,longseqcontext]{logic}
\usepackage[prefixflatinterpret,prefixflatbracketmodalinterpret,differentialdL,simplenames]{dL}

\graphicspath{ {./Images/} }

\usepackage{xspace}
\usepackage{xifthen}

\newtheorem{theorem}{Theorem}[section]

\newtheorem{corollary}[theorem]{Corollary}
\newtheorem{lemma}[theorem]{Lemma}
\newtheorem{definition}[theorem]{Definition}

\theoremstyle{definition}
\newtheorem{example}[theorem]{Example}
\newtheorem{remark}[theorem]{Remark}

\newcommand{\R}{\mathbb{R}}

\newcommand{\N}{\mathbb{N}}
\newcommand{\Q}{\mathbb{Q}}
\newcommand{\V}{\mathbb{V}}

\newcommand{\F}{\mathcal{F}}
\newcommand{\G}{\mathcal{G}}

\newcommand{\IQ}{\mathbb{I}\mathbb{Q}}
\newcommand{\A}{\mathcal{A}}

\renewcommand{\L}{\mathcal{L}}
\renewcommand{\S}{\mathbb{S}}

\newcommand{\proj}{\text{Proj}}

\renewcommand{\phi}{\varphi}

\newcommand{\rcf}{\mathcal{T}_{\text{RCF}}}

\newcommand{\theory}{\text{Th}}
\newcommand{\theoryb}{\theory_{\text{B}}}

\newcommand{\theorybco}{\theoryb^{\succeq}}
\newcommand{\eps}{\varepsilon}
\newcommand{\eval}[1]{\left\llbracket #1 \right\rrbracket}
\newcommand{\evalb}[1]{\left\llbracket #1 \right\rrbracket_{D}}
\newcommand{\norm}[1]{\left\lVert#1\right\rVert}
\newcommand{\bebecomes}{\mathrel{::=}}
\newcommand{\alternative}{~|~}

\newcommand{\pert}{\A}
\newcommand{\rels}{\{\geq, >\}}
\newcommand{\lo}[1]{\underline{#1}}
\newcommand{\up}[1]{\overline{#1}}

\newcommand{\pertf}[3][]{
    \ifthenelse{\isempty{#1}}
    {\pert^{\forall}\!\left(#2, #3\right)}
    {\pert^{\forall}_{#1}\left(#2, #3\right)}
}
\newcommand{\perte}[3][]{
    \ifthenelse{\isempty{#1}}
    {\pert^{\exists}\!\left(#2, #3\right)}
    {\pert^{\exists}_{#1}\left(#2, #3\right)}
}

\newcommand{\folr}{\text{FOL}_{\R}}
\newcommand{\term}{\text{Term}}
\newcommand{\termf}{\term_{\text{F}}}
\newcommand{\fml}{\text{Fml}}
\newcommand{\fmlb}{\fml_{\text{B}}}
\newcommand{\fmlo}{\fml^{>}}
\newcommand{\fmlc}{\fml^{\geq}}
\newcommand{\fmlco}{\fml^{\succeq}}
\newcommand{\fmlbo}{\fmlb^{>}}

\newcommand{\fmlbco}{\fmlb^{\succeq}}
\newcommand{\sent}{\text{Sent}}
\newcommand{\sentb}{\sent_{\text{B}}}

\newcommand{\sentco}{\sent^{\succeq}}

\newcommand{\sentbco}{\sentb^{\succeq}}
\newcommand{\jump}[1]{\mathbf{0}^{(#1)}}
\newcommand{\etermA}{e}

\newcommand{\fvarA}{\phi}
\newcommand{\fvarB}{\psi}

\newcommand{\funcsym}{h}
\newcommand{\interpfunc}[2][\funcsym]{{#1}_{{\ll}#2{\gg}}}

\newcommand{\leftrule}{L}
\newcommand{\rightrule}{R}

\newsavebox{\Rval}%
\sbox{\Rval}{$\scriptstyle\mathbb{R}$}
\newsavebox{\backiterateb}%
\sbox{\backiterateb}{$\scriptstyle\overleftarrow{\dibox{{}^*}}$}

\newcommand{\rfvar}{P}

\newcommand{\D}[1]{#1'}
\begin{document}

\renewcommand{\linferPremissSeparation}{\hspace{0.8cm}}

\title{Approximate Axiomatization for Differentially-Defined Functions}
\author{Andr\'e Platzer}
\author{Long Qian}
\email{platzer@kit.edu}
\email{longq@andrew.cmu.edu}
\orcid{0000-0001-7238-5710}
\orcid{0000-0003-1567-3948}
\affiliation{%
  \institution{Karlsruhe Institute of Technology}
  \city{Karlsruhe}
  \country{Germany}
  }
\affiliation{%
\institution{Carnegie Mellon University}
\city{Pittsburgh}
\country{USA}
}

\begin{abstract}
    This article establishes a \emph{complete approximate axiomatization} for the real-closed field $\R$ expanded with all differentially-defined functions, including special functions such as $\sin(x), \cos(x), e^x, \dots$. Every true sentence is provable up to some numerical approximation, and the truth of such approximations \emph{converge} under mild conditions. Such an axiomatization is a fragment of the axiomatization for \emph{differential dynamic logic}, and is therefore a finite extension of the axiomatization of real-closed fields. Furthermore, the numerical approximations approximate \emph{formulas} containing special function symbols by $\folr$ formulas, improving upon earlier decidability results only concerning closed \emph{sentences}. 
\end{abstract}



\maketitle
\section{Introduction}

Tarski's proof \cite{Tarski_1948} that the theory of real-closed fields $\rcf$ admits quantifier elimination is a seminal result, establishing a complete axiomatization for $\rcf$ and thereby providing \emph{complete symbolic reasoning principles} for first-order formulas of real arithmetic ($\folr$). This has seen many practical implementations \cite{DBLP:conf/cpp/KosaianTP23, DBLP:conf/cade/McLaughlinH05, Mathematica} and applications \cite{DBLP:journals/jsc/Jirstrand97, DBLP:conf/icstcc/VobetawinkelR19}. Extensions beyond polynomial arithmetic including standard special functions such as $\sin(x), \cos(x), e^x, \dots$ would find even more applications, but are also fundamentally challenging. It is well-known that $\R$ expanded with trignometric functions can encode integer arithmetic \cite{DBLP:journals/jsyml/Richardson68, Wang_1974}, so there cannot be any complete, computable axiomatization of the resulting theory by G\"odel's incompleteness theorem \cite{Godel}. 

Nonetheless, this article crucially establishes a complete approximate axiomatization for the bounded\footnote{Quantifiers are restricted to compact intervals.} theory of $\R$ expanded with such special functions, where every true sentence is provable up to numerical approximations. Furthermore, approximate completeness of the standard axiomatization of \emph{differential dynamic logic} ($\dL$ \cite{DBLP:journals/jar/Platzer08, DBLP:conf/cade/GallicchioTMP22}) is established, thereby providing the desired \emph{complete symbolic reasoning principles} with a natural finite extension of $\rcf$. In addition, similar to how the decidability of a theory follows from a (computable) complete axiomatization, the approximate axiomatization established also directly induce \emph{proof-generating} approximate decision procedures, extending upon earlier works on the approximate decidability of extensions of $\R$ \cite{DBLP:conf/lics/GaoAC12, Gao_Kong_Clarke_2013-dreal, Franek_Ratschan_Zgliczynski_2011, Franek_Ratschan_Zgliczynski_2016}.  

Intuitively, the axiomatization is motivated by the observation that for a function $f: \R \to \R$, the use of finite precision arithmetic implies that $f(x)$ is generally \emph{approximated} rather than evaluated \emph{exactly}. As such, one can view $f(x)$ as an interval $[\lo{f}, \up{f}]$ where $\lo{f} \leq f(x)\leq \up{f}$ and $\up{f} - \lo{f} < \eps$ is sufficiently small representing a bound on the error of the approximation. In this setting, what should formulas involving $f(x)$ such as $f(x) > 0 \land f(x) < 1$ mean? One possible interpretation is to require $[\lo{f}, \up{f}] \subset (0, 1)$, i.e. \emph{every} possible value of $f(x)$ satisfies the formula. Dually, another possible interpretation is that there exists \emph{some} $w \in [\lo{f} , \up{f}]$ such that $w$ satisfies the formula. Such inexactness crucially lowers the complexity of the theory of $\R$ expanded with special functions, resulting in an approximate axiomatization. In other words, the fundamental obstruction to achieving a complete axiomatization is the \emph{exactness} of formulas/sentences, which is no longer present once some notion of approximation is allowed. 

To illustrate, suppose a sentence $\phi = \exists x{\in} [0, 1]~f(x)\geq 0$ is given where $f(x)$ is some function. For \emph{any} $\delta \in \Q^+$ representing the approximation error tolerable, a syntactic ``$\delta$-perturbation'' of $\phi$, denoted $\perte{\phi}{\delta}$, can be constructed such that the truth of $\perte{\phi}{\delta}$ is \emph{provable} from $\dL$'s axiomatization. For sentences such as $\phi$ that only use one of $\rels$ as its relation symbol (such sentences are also called \emph{pure sentences}), the truth value of the approximation \emph{converges} to the truth value of $\phi$ as $\delta \rightarrow 0$. Such an approximation intuitively corresponds to replacing the term $f(x)$ by some $\folr$ definable approximation $p : \R \to \R$ with $\norm{f - p}_{[0, 1]} \leq \delta$ and constructing 
\[\perte{\phi}{\delta} = \exists x{\in}[0, 1]~\exists w~(\abs{w - p(x)} \leq \delta \land w \geq 0)\]
in other words, the value $f(x)$ is ``replaced'' by the interval approximation $[p(x) - \delta, p(x) + \delta]$. Moreover, the \emph{approximation procedure is provable}, because the implication $\phi \rightarrow \perte{\phi}{\delta}$ is also provable in $\dL$'s axiomatization.
Furthermore, $\dL$ is \emph{approximately complete}: for any $\delta \in \Q^+$, if $\phi$ is true, then the $\delta$-approximation $\perte{\phi}{\delta}$ is provable, if $\perte{\phi}{\delta}$ is false, then $\neg\phi$ is provable. For emphasis, standard completeness will sometimes be called \emph{exact} completeness to emphasize that completeness holds unconditionally without the presence of perturbations.
A dual construction of universal $\delta$-approximations $\pertf{\phi}{\delta}$ is also possible, corresponding to requiring $\phi$ to hold at \emph{every} possible value in the approximated interval (Definition \ref{def: perturbations}).

The main difficulty in establishing the axiomatizability of such theories comes from the inherently numerical nature of such functions, seemingly contradictory to the deductive reasoning provided by symbolic logic. For example, it is far from obvious how the bound $\norm{p - f} \leq \delta$ can be \emph{deductively proven} for a general Type-Two computable function (Definition \ref{def: computable function}) $f$. Indeed, if $f$ is only known to be a general Type-Two computable function, then properties of $f$ can only be inferred from its definition, an oracle machine that computes arbitrarily accurate numerical approximations of $f$. Considering special functions as general Type-Two computable functions, thus, takes an oracle approach: for every special function considered, a corresponding defining oracle needs to be provided and trusted, and it is unclear how to prove properties of such oracles. Furthermore, this is different from the way many functions are intuitively understood. For example, consider the trignometric function $\sin(x)$ which has multiple equivalent definitions: Taylor series, solution to ODEs, geometrically, etc. But it is unusual to \emph{define} $\sin(x)$ by some oracle machine that computes arbitrarily accurate rational approximations of $\sin(x)$ for each $x$.

This article shows that by, instead, characterizing special functions \emph{symbolically} through their defining differential equations, $\dL$ is capable of proving \emph{all} required numerical properties of \emph{all} such special functions. 
In particular, there is \emph{one} common axiomatization for \emph{all} special functions, rather than: \emph{for each} special function, \emph{there is} a different axiomatization specific to that particular special function.

This article considers the class of \emph{differentially-defined} functions. Intuitively, these are univariate functions that can be defined using polynomial differential equations \cite{DBLP:conf/cade/GallicchioTMP22}, which are very general (in fact universal \cite{Pouly_Bournez_2020}) and include all standard special functions such as $\sin(x), \cos(x), e^x, \dots$. More concretely, a differentially-defined function is the (coordinate-projected) solution to a system of polynomial ODEs. Let $p(x) \in \Q^n[x]$ be a polynomial vector field, $x_0 \in \Q^n$ a rational vector representing initial conditions, then the corresponding initial-value problem (IVP) is
\begin{align*}
    &x' = p(x)\\
    &x(0) = x_0
\end{align*}
where $\Phi(t)$ denotes the corresponding solution to the (IVP) above, and differentially-defined functions are of the form $\Phi_1(t)$, the first component of $\Phi$. 

Despite their generality, differentially-defined functions are pleasantly amenable to deductive reasoning. Recent work \cite{platzer2024axiomatizationcompactinitialvalue} has shown that $\dL$ is surprisingly powerful in proving numerical properties of solutions to initial value problems, essentially most topologically open properties of such functions with at most one quantifier can be proven in $\dL$. Crucially extending upon such prior results which only concern the fragment of $\dL$ with at most one quantifier, this article establishes approximate completeness allowing for arbitrary (bounded) quantifiers by \emph{deductively proving} the inherently numerical properties using symbolic logic in a uniform fashion.

This article further establishes \emph{robustness} (Definition \ref{def: robustness}) properties for the class of pure formulas. A formula $\phi$ is said to be \emph{pure} if, in prenex normal form, it only contains $\succeq$ as a relation symbol where $\succeq~\in \rels$, in which case it is also said to be $\succeq$-pure. Robustness of a formula can be intuitively understood as the property that the validity of the formula remains unchanged by sufficiently small numerical perturbations. Pure formulas are robust as they are topologically tame, always defining open/closed sets. In particular, for robust sentences $\phi$, it is always the case that the truth value of its approximations converge to the truth value of $\phi$ as the level of perturbation $\delta \rightarrow 0$. Hence, it follows (Theorem \ref{thm: atomic formulas are robust}) that $\dL$'s axiomatization is \emph{exactly complete} for all $>$-pure sentences. In particular, if a $>$-pure sentence $\phi$ is true, then it is provable $\vdash_{\dL} \phi$, \emph{exactly without perturbations}. 

While the results presented in this article are theoretically motivated, they also have practical implications. On the level of decidability, approximate axiomatizability directly implies the existence of \emph{proof-producing} approximate decision algorithms. For example, $\delta$-decidability has seen many applications in practice \cite{Gao_Avigad_Clarke_2012, Gao_Kong_Clarke_2013-dreal, Gao_Kong_Chen_Clarke_2014}. The usage of such procedures in safety-critical applications calls for the need to have a higher-level of trust on the results, such as generating verifiable symbolic proofs certifying correctness in some sound logical system. The axiomatization provided by this article directly gives such proof-producing decision algorithms utilizing the sound axiomatization of $\dL$. In particular, the logical system ($\dL$) is \emph{not dependent on the functions considered}, there exists a finite list of axioms that works for \emph{all} differentially-defined functions. In contrast with earlier approaches \cite{DBLP:journals/jar/AkbarpourP10, Gao_Kong_Clarke_2014} where \emph{for every function}, corresponding axioms capturing its numerical properties need to be added and \emph{trusted}, resulting in logical systems where soundness is difficult to verify. 

Furthermore, the results presented by this article provide a way to \emph{provably} approximate \emph{formulas} by quantifier-free $\folr$ formulas, a notion of \emph{approximate quantifier-elimination}. As an example application, consider the problem of searching for Lyapunov functions parametrized by the symbolic variables $\vec{c}$. Necessary conditions for the validity of such Lyapunov functions can be expressed by a formula $\phi(\vec{c})$ involving special functions (Example \ref{example: lyapunov function}). Crucially utilizing the axiomatization established, one can provably compute some quantifier-free $\folr$ formula $\psi(\vec{c})$ such that $\psi(\vec{c}) \rightarrow \phi(\vec{c})$ is (provably) valid. In other words, this proves the correctness of \emph{every} $\vec{c}$ satisfying $\psi(\vec{c})$ at once with a single proof (Example \ref{example: Lyapunov function approx}). This is beyond the capability of earlier works on approximate decidability \cite{DBLP:conf/cav/KongSG18, DBLP:conf/hybrid/KapinskiDSA14} which can only decide the truth of closed sentences without free-variables. 

This article establishes the following technical results: 

Let $\R_D$ denote the real-closed field $\R$ expanded with all differentially-defined functions, which in particular includes standard functions such as $\sin(x), \cos(x), e^x$ and many more. Then $\dL$ provides a natural approximately complete axiomatization for the bounded theory of $\R_D$, where quantifiers are restricted to compact intervals. Let $\phi$ be any bounded formula in the language of $\R_D$, then the following hold:

\begin{enumerate}
    \item For every $\delta \in \Q^+$, the $\delta$-approximations $\pert^\forall(\phi, \delta), \pert^\exists(\phi, \delta) \in \folr$ can be computed. These approximations have the same free-variables as $\phi$, and the chain of implications
    \[\pert^\forall\!(\phi, \delta)\rightarrow \phi \rightarrow \pert^\exists\!(\phi, \delta)\]
    is provable in $\dL$. Note that since $\rcf$ admits proof-producing decision procedures \cite{DBLP:conf/cade/McLaughlinH05}, this naturally induces a proof-producing $\delta$-decidability procedure \cite{DBLP:conf/lics/GaoAC12} in the following sense: If $\phi$ is a sentence, then the $\delta$-approximations are also sentences in $\folr$, and can therefore be decided with corresponding proofs. Furthermore, the following hold:
    \begin{itemize}
        \item If $\pert^\forall(\phi, \delta)$ is true, then $\phi$ is provably true (in $\dL$).
        \item If $\pert^\exists(\phi, \delta)$ is false, then $\phi$ is provably false (in $\dL$). 
    \end{itemize}
    
    \item If $\phi$ is $>$-pure, then it is $\forall$-robust (Definition \ref{def: robustness}) and is the limit of a sequence of $\folr$ formulas that are \emph{decreasing} in strength. For any bounded set $D \subset \R^n$, its universal perturbations $\pertf{\phi}{\delta}$ converge in truth:
    \[\eval{\phi} \cap D = \bigcup_{\delta > 0} \eval{\pert^\forall\!(\phi, \delta)} \cap D\]
    where $\eval{\phi} \subseteq \R^n$ denotes the semantics of a formula $\phi$, i.e. the set of all true states. Similarly, if $\phi$ is $\geq$-pure, then it is $\exists$-robust (Definition \ref{def: robustness}) and is the limit of a sequence of $\folr$ formulas that are \emph{increasing} in strength.
    \[\eval{\phi} \cap D = \bigcap_{\delta > 0} \eval{\pert^\exists\!(\phi, \delta)} \cap D\]
    So for a pure formula $\phi$, (at least one of) the $\delta$-approximations converge to $\phi$ as $\delta \to 0$. 
\end{enumerate}

Finally, this article also investigates the extent to which robustness properties can be generalized. It is shown that there are fundamental computability-theoretic obstacles to extending such robustness properties to all sentences with bounded quantifiers when the function symbols are allowed to be arbitrary Type-Two computable functions, by establishing that the theory is at least as complicated as arithmetic in the arithmetic hierarchy. 

\section{Related Works}
This article presents an approximately complete axiomatization for $\R$ expanded with differentially-defined functions. Furthermore, the presented axiomatization is a finite fragment of $\dL$ and therefore a finite extension of $\rcf$. To the best of our knowledge, this article is the first to consider such notions of approximate axiomatizability that works uniformly for all differentially-defined functions.
\\
\\
\noindent
{\bf{\emph{Approximate Decidability and Perturbations}}}: Previous works concerning perturbations \cite{Gao_Avigad_Clarke_2012, DBLP:journals/jsc/Ratschan02, Franek_Ratschan_Zgliczynski_2016} focused on notions of approximate decidability of sentences, and therefore do not give the desired symbolic reasoning principles for \emph{formulas} that an axiomatization provides. In contrast, an axiomatization automatically provides corresponding \emph{proof-producing} decidability results.

Indeed, generalizing the decidability results of earlier works to an axiomatization appears to be challenging, since earlier works primarily considered extending $\R$ with \emph{all} Type-Two computable functions (Definition \ref{def: computable function}). Deductively proving properties of such functions requires essentially proving properties of the underlying oracle machines that define these functions, and it is not clear how this could be achieved. Earlier attempts at proof-generation \cite{Gao_Kong_Clarke_2014, DBLP:journals/jar/AkbarpourP10} essentially assume the desired numerical properties of such functions as \emph{axioms}, and the validation of such axioms are offloaded to external numerical algorithms. Hence the correctness of the generated proofs and the soundness of the underlying logical system is only relative to these intricate axioms, whose correctness is far from clear and needs to be validated by external ``black-box algorithms'', while also preventing the possibility of an axiomatization that is a finite extension of $\rcf$. 

The results of this article present an alternate viewpoint - such earlier results are difficult to axiomatize partly because the language of all Type-Two computable functions is \emph{unnecessarily general}, lacking clear symbolic representations. The class of differentially-defined functions \cite{DBLP:conf/cade/GallicchioTMP22} is a more suitable expansion of $\R$ for such purposes (denoted $\R_D)$. Intuitively, such functions are the (coordinate-projected) solutions of polynomial ordinary differential equations (ODEs) which are very general and universal among all continuous functions \cite{Pouly_Bournez_2020}. Despite the generality of $\R_D$, this article establishes a natural and approximately complete axiomatization using $\dL$. Thus, the soundness of the axiomatization automatically follows from the soundness of $\dL$ \cite{DBLP:journals/jar/Platzer17}, which has been formally verified using various approaches \cite{DBLP:journals/corr/abs-2404-15214, DBLP:conf/cpp/BohrerRVVP17, Differential_Dynamic_Logic-AFP}. Consequently, there is a finite list of sound axioms that capture all necessary deductive reasoning principles for \emph{all} differentially-defined functions, as opposed to earlier works \cite{Gao_Kong_Clarke_2014,DBLP:journals/jar/AkbarpourP10,Franek_Ratschan_Zgliczynski_2016} where for each special function, potentially unsound axioms need to be added and trusted. On a more practical level, such expanded functions are now more naturally defined by their \emph{differential equations}. E.g. $\sin(x)$ is succinctly defined by the ODE: $x' = y, y' = -1, x(0) = 0, y(0) = 1$. In contrast, if $\sin(x)$ were treated as a Type-Two computable function, it would then be represented by some oracle machine that outputs arbitrarily approximations of $\sin(x)$ using names (Definition \ref{def: computable function}), resulting in a much more complicated definition that is difficult to check for correctness. 

In the specific case of satisfiability checking, various practical algorithms have been proposed \cite{DBLP:journals/jar/AkbarpourP10,tocl_smt_transcendental,jar_smt_search} to decide exact satisfiability. Such algorithms also replace special functions with appropriate polynomial estimates, but the correctness of such approximations are often only proven semantically rather than deductively in a single logical system. Thus, the soundness of such algorithms depend on the soundness of such approximation axioms which are delicate to justify and is dependent on the special functions considered, a new axiom is needed for each new approximation of each special function. The general axiomatization provided by this article offers a complementary approach, that such delicate arguments can be carried out deductively and uniformly for \emph{all} differentially-defined functions at once utilizing the logic $\dL$ for the \emph{full bounded theory}, resulting in trustworthy proofs of correctness in a single sound logical system.

\noindent
{\bf{\emph{Proof theory of ODEs}}}: As differentially-defined functions are solutions of ODEs, the results presented in this article extend upon earlier works concerning the proof theory of ODEs \cite{DBLP:conf/lics/Platzer12a, DBLP:journals/jacm/PlatzerT20, platzer2024axiomatizationcompactinitialvalue}, which is relatively unexplored compared to the computability/complexity of theory of ODEs \cite{Pouly_Bournez_2020, DBLP:journals/tocl/Ratschan06, Graca_Zhong_Buescu_2009, DBLP:journals/jc/BournezGGP23, Gao_Avigad_Clarke_2012, DBLP:journals/toct/KawamuraC12}. Earlier works only consider either \emph{exact} completeness \cite{DBLP:journals/jacm/PlatzerT20, platzer2024axiomatizationcompactinitialvalue} without approximations or completeness relative to some non-computable oracle \cite{Platzer_2012}, and therefore do not provide computable axiomatizations that are approximately complete. In particular, earlier works that give computable axiomatizations only concern the fragment of $\R_D$ with at most one quantifier \cite{platzer2024axiomatizationcompactinitialvalue}, whereas the results in this article allow for arbitrary (bounded) quantifiers and are therefore much more general.

\section{Preliminaries}
\label{sec: preliminaries}
This section provides a self-contained overview of the prerequisites required. More details on computable analysis \cite{Weihrauch_2000,Pour-El_Richards_2017}, computability theory \cite{Soare_2016} and differential dynamic logic \cite{DBLP:journals/jar/Platzer08, DBLP:conf/cade/GallicchioTMP22} can be found in standard references. 

\subsection{Differential Dynamic Logic}
\label{sec: prelim_dL}
This section provides a brief review of differential dynamic logic ($\dL$) and its axiomatization, focusing on the continuous fragment of $\dL$.

\subsubsection{Syntax}
Terms in $\dL$ are formed by the following grammar, where $\V$ denotes the set of all variables, $x \in \V$ is a variable, $h$ a $k$-ary function symbol and $c \in \Q$ is a rational constant.
\[p, q\bebecomes x\alternative c\alternative p + q\alternative p\cdot q \alternative h(p_1, \cdots, p_k)\]
$\dL$ formulas have the following grammar, where ${\sim} \in \{=, \neq, \geq, >, \leq, <\}$ is a comparison relation and $\alpha$ is a system of differential equations ($\dL$ allows for $\alpha$ to be from the more general class of \emph{hybrid programs} \cite{DBLP:journals/jar/Platzer08}, which is not needed here)
\begin{align*}
    \phi, \psi &\bebecomes p {\sim} q \alternative \phi \land \psi \alternative \phi \lor \psi \alternative \neg \phi \alternative \forall x \phi \alternative \exists x \phi \alternative 
\ddiamond{\alpha}{\phi} \alternative \dbox{\alpha}{\phi}\\
\alpha &\bebecomes \cdots \alternative x' = f(x) \& Q
\end{align*}
Intuitively, the modal formulas $\ddiamond{\alpha}{\phi}, \dbox{\alpha}{\phi}$ expresses \emph{liveness} and \emph{safety} properties. $\ddiamond{\alpha}{\phi}$ expresses that by following the ODE $\alpha$, there is some time for which $\phi$ is true. Dually, $\dbox{\alpha}{\phi}$ expresses that $\phi$ is always true following the ODE $\alpha$. In this paper, we will only be dealing with the case $\alpha \equiv x' = f(x) \& Q$, where $x' = f(x)$ represents an autonomous system of ODEs $x_1' = f_1(x), \cdots, x_n' = f_n(x)$, $x = (x_1, \cdots, x_n)$ is understood to be vectorial and $Q$ is some $\folr$ formula known as the \emph{domain constraint}. Intuitively, this restricts the region for which the ODE $x' = f(x)$ is allowed to evolve.

The function symbols $h$ are in general \emph{uninterpreted} in $\dL$'s uniform substitution calculus \cite{DBLP:journals/jar/Platzer17}, semantically corresponding to arbitrary $C^\infty$ functions. To allow for computable reasoning and remain in a countable language, function symbols $h$ in this paper carry an \emph{interpretation annotation}, $\interpfunc{\phi}$, where $\phi(x_0, y_1, \cdots, y_k)$ is a $\dL$ formula with no uninterpreted symbols. Intuitively, $\phi$ defines the graph of some $C^\infty$ function intended as the interpretation of $h$. This paper considers the class of univariate, differentially-definable functions \cite{DBLP:conf/cade/GallicchioTMP22}. In which case the formulas $\phi(x_0, y)$ have the following shape, where $x = (x_0, \cdots, x_n)$ denotes a vector variables, $t = y_1$ is the univariate input, $f(x) = (f_0(x), \cdots, f_n(x)) \in \Q^{n + 1}[x]$ is a polynomial vector field with rational coefficients, $X = (X_0, \cdots, X_N)$ and $T$ are rational constants
{\small\begin{align}
    \fvarA(x_0,t) &\mnodefequiv \exists x_1, \dots, x_n\left(\ddiamond{\pevolve{\D{x}=-\genDE{x},\D{t}=-1}}{\left(x = X \land t = T\right)} \lor \ddiamond{\pevolve{\D{x}=\genDE{x},\D{t}=1}}{\left(x = X \land t = T\right)}\right)
    \label{eq:ddef}
\end{align}}

Such formulas $\phi(x_0, t)$ define corresponding univariate functions $\hat{h}(t)$ where $\hat{h}(t)$ is the first coordinate projection of the solution to the initial value problem (IVP)
\begin{align*}
    &x' = f(x)\\
    &x(T) = X
\end{align*}
For function $\hat{h}(t)$ to be semantically well-defined, it is necessary for the IVP above to admit \emph{global solutions}, which is equivalent to the validity of the formula $\exists t \phi(x_0, t)$. Hence, this paper only considers function symbols $\interpfunc{\phi}$ where the corresponding $\dL$ formula $\exists t \phi(x_0, t)$ is provable. By prior work on the provability of global existence, standard mathematical functions such as $\sin(t), \cos(t), e^t$ as well as functions that admit definable upper-bounds in $\dL$ all satisfy this criteria \cite{DBLP:journals/fac/TanP21}.

For terms and formulas that appear in contexts involving ODEs $x' = f(x)$, it is sometimes needed to restrict the variables that they can refer to. Maximal permitted free variables will be indicated by explicitly writing them as arguments when needed. For example, $p()$ means that the term $p$ cannot refer to any bound variable of the ODE $x' = f(x)$. In contrast, $P(x)$ (or just $P$) indicates that all the variables may be referred to as free variables. Such variable dependencies can be made formal and rigorous through $\dL$'s uniform substitution calculus \cite{DBLP:journals/jar/Platzer17}.

\subsubsection{Semantics}
A state $\omega$ is a mapping $\omega : \V \to \R$ that assigns a value to every variable. We denote $\S$ as the set of all such states. For a term $e$, its semantics in state $\omega \in \S$ written as $\eval{e}$ is the real value obtained by evaluating the $e$ at the state $\omega$. More specifically, for a function term $\interpfunc{\phi}(e)$, its semantics is defined as 
\[\eval{\interpfunc{\phi}(e)} = \hat{h}(\eval{e})\]
where $\hat{h}$ is the function defined by $\phi(x_0, t)$.  
For a $\dL$ formula $\phi$, its semantics $\eval{\phi}$ is defined to be the set of all states $\omega \in \S$ such that $\omega \models \phi$, i.e the formula $\phi$ is true in $\omega$. The semantics of first-order logical connectives are defined as expected, e.g. $\eval{\phi \lor \psi} = \eval{\phi} \cup \eval{\psi}$. For ODEs $\alpha \equiv x' = f(x) \& Q$, the semantics for $\dbox{\alpha}{\phi}$ and $\ddiamond{\alpha}{\phi}$ are defined as follows. For the given ODE $\alpha$ and any state $\omega \in \S$, let $\Phi_\omega : [0, T) \to \S$ be the solution to $x' = f(x)$ extended maximally to the right with $0 < T \leq \infty$ and $\Phi_\omega(0) = \omega$. We then have:
\begin{align*}
    &\omega \models \dbox{\alpha}{\phi} \text{ iff for all $0 \leq \tau < T$ such that $\Phi_\omega(\xi) \models Q$ for all $0 \leq \xi \leq \tau$, we have $\Phi_\omega(\tau) \models \phi$}\\
    &\omega \models \ddiamond{\alpha}{\phi} \text{ iff there exists some $0 \leq \tau < T$ such that $\Phi_\omega(\xi) \models Q$ for all $0 \leq \xi \leq \tau$ and $\Phi_\omega(\tau) \models \phi$}
\end{align*}
Intuitively, the formula $\dbox{\alpha}{\phi}$ expresses a \emph{safety} property, that $\phi$ holds along all flows of the ODE $x' = f(x)$ that remain inside the domain constraint $Q$. Similarly, the formula $\ddiamond{\alpha}{\phi}$ expresses a \emph{liveness} property, that there is some flow along $x' = f(x)$ staying within $Q$ eventually reaching a state where $\phi$ is true. Finally, a formula $\phi$ is said to be valid if $\eval{\phi} = \S$, i.e. it is true in all states. 

\subsubsection{Proof calculus}
The derivations in this paper are presented in a standard, classical sequent calculus with the usual rules for manipulating logical connectives and sequents. For a \emph{sequent} $\Gamma \vdash \phi$, its semantics is equivalent to the formula $\left(\bigwedge_{\psi \in \Gamma} \psi\right) \rightarrow \phi$, and the sequent is called valid if its corresponding formula is valid. For a sequent $\Gamma \vdash \phi$, formulas $\Gamma$ are called antecedents, and $\phi$ the succedent. Completed proof branches are marked with $*$ in a sequent proof, and since $\R$ has a decidable theory via quantifier elimination \cite{Tarski_1948}, statements that follow from real arithmetic are proven with the rule \irref{qear}. An axiom (schema) is called \emph{sound} iff all of its instances are valid, and a proof rule is sound if the validity of all its premises entail the validity of its conclusion. Axioms and proof rules are \emph{derivable} if they can be proven from $\dL$ axioms and proof rules via the aforementioned sequent calculus. Derivable axioms are automatically sound due to the soundness of $\dL$'s axiomatization \cite{DBLP:journals/jar/Platzer08, DBLP:journals/jacm/PlatzerT20}.

This paper uses a fragment of the base axiomatization of $\dL$ \cite{DBLP:conf/lics/Platzer12b} (focusing on the continuous case) along with an extended axiomatization developed in prior works used to handle ODE invariants and liveness properties \cite{DBLP:journals/jacm/PlatzerT20, DBLP:journals/fac/TanP21}. A complete list of the axioms used is provided in \rref{app: dL axiomatization}. The following function interpretation axiom relates differentially-defined functions to their annotations. 

\begin{theorem}[\cite{DBLP:conf/cade/GallicchioTMP22}]
The~\irref{FI} axiom (below) for $\dL$ is sound where $\funcsym$ is a $k$-ary function symbol and the formula semantics $\eval{\phi}$ is the graph of a smooth $C^\infty$ function $\hat{\funcsym} : \reals^k \to \reals$.

\noindent
\[
\cinferenceRule[FI|FI]{function interpretation}
{\linferenceRule[equiv]
  {\fvarA(\etermA_0,\etermA_1,\dots,\etermA_k)}
  {\etermA_0 = \interpfunc[\funcsym]{\fvarA}(\etermA_1,\dots,\etermA_k)}
}
{}%
\]
\end{theorem}

Consequently, axiom \irref{FI} is sound for all function symbols appearing in this paper, which only considers $\phi(x_0, t)$ of the form in (\ref{eq:ddef}) and $\exists t \phi(x_0, t)$ provable. While technically there are infinitely many instances of the \irref{FI} axiom resulting in an infinite axiomatization of $\R_D$, this is only needed to naturally interpret formulas of $\R_D$ in $\dL$. If desired, one could also soundly transform formulas of $\R_D$ to $\dL$ formulas without special functions using the soundness of \irref{FI}. 

It was recently established that $\dL$ is complete for $\folr$ approximations \cite[Theorem~5.1]{platzer2024axiomatizationcompactinitialvalue}. In particular, the classical Stone-Weierstrass approximation theorem can be deductively proven in $\dL$. This helpful result also establishes that arbitrarily accurate provable approximations of differentially-defined functions can be always computed, as shown below. 

\begin{theorem}[Stone-Weierstrass {\cite[Theorem~5.5]{platzer2024axiomatizationcompactinitialvalue}}]
\label{thm: stone weierstrass for IVPs}
    Let $f(x) \in \Q^n[x]$ be a polynomial vector field with rational coefficients, $X \in \Q^n$ rational initial conditions, $t_0, T \in \Q$ a rational interval with $t_0 \leq T$. Further suppose that the solution to initial value problem
    \begin{align*}
        &x' = f(x)\\
        &x(t_0) = X
    \end{align*}
    does not exhibit finite time blow-up on $[t_0, T]$. Then for all $\eps \in \Q^+$, there (computably) exists some $\theta(t) \in \Q^n[t]$ such that the following formula is provable in $\dL$:
    \[x = X \land t = t_0 \rightarrow \dbox{x' = f(x), t' = 1 \& t \leq T}{\norm{x - \theta(t)}^2 \leq \eps^2}\]
\end{theorem}
In other words, for all desired errors $\eps \in \Q^+$, one can compute a corresponding polynomial approximation $\theta$ to the solution of initial value problem (IVP) $x' = f(x), x(t_0) = X_0$ with provable uniform error at most $\eps$. For the purposes of this article, this implies that $\dL$ is capable of (computably) proving arbitrarily accurate polynomial approximations for differentially defined functions. 

\begin{theorem}[Stone-Weierstrass for Differentially-Defined Functions]
    \label{thm: stone weierstrass for functions}
    Let $h(t)$ be a differentially-defined function, and $T \in \Q^+$ be arbitrary. Then for all $\eps \in \Q^+$, there (computably) exists some $p(t) \in \Q[t]$ such that the following formula is provable in $\dL$:
    \begin{equation}
        \label{fml: error bound}
        -T \leq t \leq T \rightarrow \abs{p(t) - h(t)} < \eps
    \end{equation}
\end{theorem}

\begin{proof}
    Since differentially-defined functions are solutions to IVPs, it suffices to apply Theorem \ref{thm: stone weierstrass for IVPs} on both the positive and negative time regions to obtain the desired approximation, see Appendix \ref{app: proofs} for a complete proof. 
\end{proof}

\subsection{Computability and Computable analysis}
\label{sec: Computable analysis prelims}
This section provides a brief overview of computable analysis under the standard framework of \emph{Type Two Theory of Effectivity} (TTE) \cite{Weihrauch_2000}.  

\begin{definition}[Name]
    Let $x \in \R$ be any real number, a $\it{name}$ for $x$ is a sequence of rationals $(q_i)_i \subseteq \Q$ such that
    \[\boldsymbol{\forall} i \in \N~(\abs{q_i - x} < 2^{-i})\]
    This definition naturally extends to $\R^n$ by requiring names to reside in $\Q^n$ and using the standard Euclidean norm. For $x \in \R^n$, we denote the set of all names of $x$ as $\Gamma(x)$.
\end{definition}

For a fixed real number $x \in \R^n$, one should think of its names as the ``descriptions'' of $x$. We then define $x$ to be computable if it exhibits a computable description. 

\begin{definition}[Type-two computable number]
    \label{def: computable real}
    Let $x \in \R^n$ be any real number, $x$ is \emph{Type-Two computable} if it has a computable name i.e. there is some computable sequence $(q_i)_i \subseteq \Q^n$ that is a name for $x$. 
\end{definition}

Intuitively, this means that a real $x \in \R^n$ is (Type-Two) computable if and only if it can be computably approximated by a sequence of vectors of rational numbers. The definition above relies on some ordering of the rationals, but any fixed effective enumeration of rationals gives equivalent notions. From now on, whenever we refer to the computability of numbers in $\R^n$, we mean Type-Two computability.

\begin{definition}
An \emph{oracle machine $M$} is a Turing machine that allows for an additional one-way read-only input tape that represents some input oracle used. The machine is allowed to read this input tape up to arbitrary, but finite, lengths.
\end{definition}

One can think of oracle machines as regular Turing machines but with some access to outside information, namely the ``oracle'' input tape. The machine may use any finite amount of information on this tape. For an oracle machine $M$, and an infinite binary sequence $p \in 2^{\omega}$, $M^p$ represents the oracle machine $M$ with oracle $p$. By standard encoding, we do not differentiate between $\Q^\omega$ and $2^\omega$. 

The following definition relates the use of oracle machines to computable functions in TTE.

\begin{definition}[Computable function]
    \label{def: computable function}
    A function $f : \R^n \to \R^m$ is \emph{computable} if there is some oracle machine $M$ such that 
    \[\boldsymbol{\forall}x \in \R^n~\boldsymbol{\forall} p \in \Gamma(x)~((M^p(i))_i \in \Gamma(f(x)))\]
    I.e. $M$ maps names of $x$ to names of $f(x)$ for all $x \in E$. 
\end{definition}
Intuitively, this means that a function $f : \R^n \to \R^m$ is computable if and only if there is some computable algorithm such that for every $x \in \R^n$, the algorithm can output more and more accurate approximations of output $f(x)$ given more and more accurate approximations of input $x$. By this definition, any Type-Two computable function is necessarily continuous, since oracle machines can only read a finite amount of its oracle before producing an output. In other words, for all $x \in \R^n, i \in \N$, there is some corresponding $j \in \N$ such that if $f$ is provided with an approximation of $x$ accurate up to $2^{-j}$, then the output is an approximation of $f(x)$ accurate up to $2^{-i}$, therefore $f$ is continuous. The standard functions $\sin(x), \cos(x), x^2, e^x, \cdots$ are all computable via their Taylor expansions.

\section{$\delta$-perturbations}
\label{sec: perturbations}
This section introduces the $\delta$-perturbations used to obtain an approximate axiomatization of $\R_D$. Earlier works on $\delta$-decidability \cite{DBLP:conf/lics/GaoAC12} perturb \emph{relations} whereas this paper introduces perturbations on \emph{(function) terms}. This has the advantage that each perturbed sentence is now decidable only using $\folr$ (first-order rather than retaining special functions). The following fixes some notations that will be used in this paper. 

\begin{definition}[Notations]
    The following standard notations will be used. 
    \begin{itemize}
        \item $\IQ$ denotes the set of all intervals with rational endpoints, i.e.
        \[\IQ = \{[\lo{a}, \up{a}] ~\vert~ \lo{a} \leq \up{a} \in \Q\}\]
        \item $\R = \langle \R, 0, 1, +, \cdot, \geq, > \rangle$ is always treated as a real-closed field. The logical symbol `$=$' is defined as $x = y \iff x \geq y \land y \geq x$. We assume without loss of generality that the only relation symbols used are $\rels$ and atomic formulas are of the form $t \succeq 0$ for $\succeq~\in \rels$ and $t$ a term. 
        \item $\R_D$ is the extension of $\R$ with all differentially-defined functions with provable global existence, $\R_C$ is the extension of $\R$ with all univariate computable functions.
        \item $\L_R$ denotes the standard language of real-closed fields, assumed to include all rationals, definable functions and the relation symbol $\geq$ without loss of generality. $\L_D$ denotes the language of $\R_D$ and $\L_C$ denotes the language of $\R_C$. 
        \item For a language $\L$, $\fml(\L)$/$\sent(\L)$ denotes the set of all $\L$ formulas/sentences and $\term(\L)$ denotes the set of all terms. For a formula $\phi \in \fml(\L)$, $\term(\phi)$ denotes all terms appearing in $\phi$ and $\termf(\phi)$ denotes all function terms (i.e. terms of the form $f(e)$). A term $t \in \term(\L)$ is \emph{function-free} if it does not contain function symbols. By a slight abuse of notation, a term $e(\vec{x}) \in \term(\L)$ is also identified with the function $\vec{x} \mapsto \eval{e(\vec{x})}$. $\fml(\L_R)$ is also denoted as $\folr$. 
        \item For $\L \in \{\L_R, \L_C, \L_D\}$, $\psi \in \fml(\L)$, $a, b \in \term(\L)$ and $Q \in \{\exists, \forall\}$ a quantifier, the bounded quantifier $Q^{[a, b]}~\psi(x, \vec{y})$ is defined as:
        \begin{align*}
            &\exists^{[a, b]}x~\psi(x, \vec{y}) = \exists x \left(a \leq x \land x \leq b \land \psi(x, y)\right)\\
            &\forall^{[a, b]}x~\psi(x, \vec{y}) = \forall x \left(a \leq x \land x \leq b \rightarrow \psi(x, y)\right)
        \end{align*}
        \item For $\L \in \{\L_R, \L_C, \L_D\}$, a formula $\phi \in \fml(\L)$ is \emph{bounded} if
        \[\phi(\vec{x}) \equiv Q_1^{I_1}y_1Q_2^{I_2}y_2\cdots Q_n^{I_n}~y_n\psi(\vec{x}, y_1, \cdots, y_n)\]
        where $Q_i \in \{\exists, \forall\}$ denotes the $i$-th quantifier, $I_i = [a_i, b_i]$ are closed intervals that could potentially depend on variables appearing before $y_i$ and $\psi(\vec{x}, y_1, \cdots, y_n)$ is a quantifier free formula in conjunctive normal form. We assume without loss of generality that all (bounded) formulas are in this normal form. 
        \item For $\L \in \{\L_R, \L_C, \L_D\}$, $\succeq~\in \rels$ $\fmlco(\L)$/$\sentco(\L)$ denotes the formulas/sentences that only contain the relation $\succeq$.
        \item For a language $\L \in \{\L_R, \L_C, \L_D\}$, the set of all bounded formulas/sentences is denoted as $\fmlb(\L)$/$\sentb(\L)$. Note that a bounded formula $Q^{[a, b]}x\psi(x, \vec{y})$ can always be equivalently written as
        \[Q^{[a, b]}x~\psi(x, \vec{y}) \equiv Q^{[0, 1]}t~\psi(at + (1 - t)b, \vec{y})\]
        Thus, bounded formulas/sentences are assumed without loss of generality to only quantify over intervals $[a, b]$ where the terms $a, b$ do not contain function symbols. The following notations will be useful (where $\succeq~\in \rels$):
        \begin{align*}
            \fmlbco(\L) &\coloneqq \fmlb(\L) \cap \fmlco(\L)\\
            \sentbco(\L) &\coloneqq \sentb(\L) \cap \sentco(\L)
        \end{align*}
        \item For $\L \in \{\L_R, \L_C, \L_D\}$ and $R$ an $\L$-structure, $\theory(R)$ denotes the theory of $R$ and $\theoryb(R)$ denotes the bounded fragment of $\theory(R)$. Restricting the theory will also be useful (where $\succeq~\in\rels)$:
        \begin{align*}
            \theorybco(R) &\coloneqq \theoryb(R) \cap \sentbco(\L)\\
        \end{align*}
    \end{itemize}
\end{definition}

\begin{remark}
    For clarity, this article generally suppresses explicit mentions of free variables, and will refer to a formula $\phi(\vec{x})$ by $\phi$. When used, the vectorial variable $\vec{x}$ always refers to the free variables of the corresponding formula which will be clear from context. 
\end{remark}

This article approximates formulas by computing \emph{provable perturbations} for general (function) terms using $\folr$ definable functions, resulting in $\folr$ approximations with the same set of free-variables. In particular, approximations of $\L_D$-sentences are $\L_R$-sentences, from which their validity can be computably decided as $\rcf$ is decidable. The following serves as a running example throughout the article. 

\begin{example}[Existence of Lyapunov Functions in {$\R_D$}]
    \label{example: lyapunov function}
    Consider the following model of a pendulum system:
    \begin{align*}
        &x_1' = x_2\\
        &x_2' = -\sin(x_1) - x_2
    \end{align*}
    Recall that a (strict) Lyapunov function is a function $V(x_1, x_2)$ such that there exists some $\gamma > 0$ satisfying:
    \begin{align*}
        &V(0) = 0\\
        &\forall \vec{x} \left(0\!<\!\norm{\vec{x}}\!<\!\gamma \rightarrow V(\vec{x})\!>\!0 \land -\dot{V}(\vec{x})\!>\!0\right)
    \end{align*}
    The existence of a (strict) Lyapunov function implies the (local) asymptotic stability of $(0, 0)$. By using a quadratic template $V(x_1, x_2) = c_1x_1x_2 + c_2x_1^2 + c_3x_2^2$, the sufficiency of $V$ being a valid Lyapunov function for the ODE system can be characterized by the following $\Pi_1$ formula in $\L_D$ (note that $V(0, 0) = 0$ by construction):
    \begin{align*}
        \phi(c_1, c_2, c_3) &\equiv \forall \vec{x} \left({0}\!<\!\norm{\vec{x}}^2\!<\!{1} \rightarrow V(x, y)\!>\!0 \land \frac{\partial V}{\partial x_1}x_2 + \frac{\partial V}{\partial x_2}(-\sin(x_1) - x_2)\!<\!0 \right)\\
        &\equiv \forall \vec{x}\left({0}\!<\!\norm{\vec{x}}^2\!<\!{1} \rightarrow V(x, y)\!>\!0 \land (c_1x_1 + 2c_3x_2)(x_2 + \sin(x_1)) -x_2(2c_2x_1 + c_1x_2)\!>\!0 \right)
    \end{align*}
    The existence of such coefficients $c_1, c_2, c_3$ can now be encoded by existentially quantifying over these variables, resulting in a $\exists\forall$ formula with the function symbol $\sin(x)$. 
\end{example}

Similar to solving IVPs numerically, approximations to differentially-defined functions are carried out on compact intervals. E.g. $e^x$ can only be arbitrarily approximated by polynomials on compact intervals and not all of $\R$. As such, a bound on the domains of special functions need to be computed, resulting in the following definitions.

\begin{definition}[Domain of bounded formulas]
    \label{def: domains}
    For a bounded formula $\phi \in \fmlb(\L_D)$, a \emph{domain} of $\phi$ is a function $D : FV(\phi) \to \IQ$. For clarity, $D(x)$ is also written as
    \[D(x) = D_x = [\lo{x}, \up{x}] \in \IQ\]
\end{definition}
Intuitively, a domain $D$ of a formula $\phi$ restricts the free variables to the bounded set $\Pi_{x \in FV(\phi)} D_x$, which is also denoted with $D$ by abuse of notation. When quantifying free variables over the entire domain via $Q^{D_{x_1}}x_1\dots Q^{D_{x_n}}x_n$, this is abbreviated with $Q^{D}\vec{x}$. The following definition extends domains to provide interval bounds for all terms appearing in $\phi$.

\begin{definition}[$\eps$-enclosures]
    \label{def: enclosures}
    Let $\eps \in \Q^+$, and $\phi \in \fmlb(\L_D)$ be a bounded formula. An $\eps$-enclosure of $\phi$ is a function $I : \term(\phi) \to \IQ$ satisfying the following requirements. For clarity, $I(e)$ for a term $e$ is also written as
    \[I(e) = I_e = [\lo{e}, \up{e}]\]

    \begin{enumerate}
        \item For $x \in \term(\phi)$ a free variable, $I_x$ enforces no requirements 
            \[I_x \in \IQ\]
        \item For $x \in \term(\phi)$ a bound variable introduced via the quantifier $Q^{[a, b]}x$, $I_x$ satisfies
        \[I_a \cup I_b \subseteq I_x\]
        \item For a term $t = t_1 + t_2 \in \term(\phi)$, $I_t$ satisfies
        \[I_{t_1} + I_{t_2} \subseteq I_{t_1 + t_2}\]
        with the usual addition of intervals:
            \[[\lo{t_1}, \up{t_1}] + [\lo{t_2}, \up{t_2}] = [\lo{t_1} + \lo{t_2}, \up{t_1} + \up{t_2}]\]
        \item For a term $t = t_1\cdot t_2 \in \term(\phi)$, $I_t$ satisfies
            \[I_{t_1} \cdot I_{t_2} \subseteq I_{t_1 \cdot t_2}\]
            with the usual multiplication of intervals:
        \[[\lo{t_1}, \up{t_1}] \cdot [\lo{t_2}, \up{t_2}] = [\min(\lo{t_1}\lo{t_2}, \lo{t_1}\up{t_2}, \up{t_1}\lo{t_2}, \up{t_1}\up{t_2}), \max(\lo{t_1}\lo{t_2}, \lo{t_1}\up{t_2}, \up{t_1}\lo{t_2}, \up{t_1}\up{t_2})]\]
        \item For a term $f(e) \in \term(\phi)$, $I_{f(e)}$ satisfies
            \[f(I_e) + [-\eps, \eps] \subseteq I_{f(e)}\]
            where $f$ is identified with its interpretation in $\R_D$ and $f(I_e)$ denotes the image of $I_e$ under $f$. 
    \end{enumerate}
    $1$-enclosures are also referred to as enclosures. An $\eps$-enclosure $I$ is \emph{provable} if the conditions above are provable in $\dL$. 
\end{definition}

Intuitively, $\eps$-enclosures for a formula $\phi$ bounds the ranges of terms appearing $\term(\phi)$. Also note that every $\eps$-enclosure naturally induces a domain. The following result shows that $\eps$-enclosures extending a given domain can always be computed. 

\begin{theorem}[Computable $\eps$-enclosures]
    \label{thm: compute enclosures}
    For all $\phi \in \fmlb(\L_D)$ and $D : FV(\phi) \to \IQ$ a domain of $\phi$, for all $\eps \in \Q^+$, an $\eps$-enclosure $I : \term(\phi) \to \IQ$ extending $D$ ($I\lvert_{FV(\phi)} = D$) can be computed uniformly in $D, \phi, \eps$.
\end{theorem}

\begin{proof}
    The desired $\eps$-enclosure $I$ is constructed by induction:
    \begin{itemize}
        \item For $x \in \term(\phi)$ a free variable: 
            \[I_x = D_x\]
        \item For $x \in \term(\phi)$ a bound variable introduced via $Q^{[a, b]}x$:
        \[I_x = [\lo{a}, \up{b}]\]
        \item For terms $t = t_1 + t_2 \in \term(\phi)$:
            \[I_{t} = I_{t_1} + I_{t_2}\]
        \item For terms $t = t_1\cdot t_2$:
            \[I_t = I_{t_1} \cdot I_{t_2}\]
        \item For terms $f(e) \in \term(\phi)$, first compute rationals $p, q \in \Q$ such that 
        \begin{align*}
            &\abs{p - \min_{s \in I_e} f(e)} < 1\\
            &\abs{q - \max_{s \in I_e} f(e)} < 1
        \end{align*}
        Note that this is computable as solutions to IVPs are computable \cite{Graca_Zhong_Buescu_2009}. Define 
        \[I_{f(e)} = [p - \eps - 1, q + \eps + 1]\]
        it follows by construction that 
        \[f(I_e) + [-\eps, \eps] \subseteq [p - 1, q + 1] + [-\eps, \eps] = I_{f(e)}\]
    \end{itemize}
    The above construction is computable in $\phi, \eps$ and yields an $\eps$-enclosure of $\phi$, which completes the proof. 
\end{proof}

Theorem \ref{def: enclosures} shows that $\eps$-enclosures can always be computed. A natural strengthening is to consider if \emph{provable} $\eps$-enclosures can also be computed. The following result proves this by showing that all $\eps$-enclosures are provable $\sfrac{\eps}{2}$-enclosures. 

\begin{theorem}[$\eps$-enclosures are provable]
    \label{thm: enclosures are provable}
    Let $\phi \in \fmlb(\L_D)$ be a bounded formula, $\eps \in \Q^+$ and $I : \term(\phi) \to \IQ$ an $\eps$-enclosure. Then $I$ is a \emph{provable} $\sfrac{\eps}{2}$-enclosure for $\phi$.
\end{theorem}

\begin{proof}
    Clearly every $\eps$-enclosure is an $\sfrac{\eps}{2}$-enclosure, so it suffices to establish the provability claims, this is proven inductively on the structural complexity of terms. Let $I$ be an $\eps$-enclosure. 
    \begin{itemize}
        \item For $x \in \term(\phi)$ a free variable, no restrictions are placed on $I_x$ and therefore there is nothing to prove. 

       \item For terms of the form $t = t_1 + t_2 \in \term(\phi)$, it suffices to notice that the condition $I_{t_1} + I_{t_2} \subseteq I_t$ is a comparison of rational intervals, therefore provable with axiom \irref{qear}. An identical argument also handles $t = t_1 \cdot t_2$ and bound variables.

       \item Let $f(e) \in \term(\phi)$ be a term and $T \in \Q^+$ be large enough such that $I_e \subseteq [-T, T]$. Let $p \in \Q[t]$ be some polynomial approximation of $f$ on the interval $[-T, T]$ with provable uniform error at most $\sfrac{\eps}{4}$ computed by Theorem \ref{thm: stone weierstrass for functions}. The following $\folr$ formula (using standard abbreviations) is then valid by the triangle inequality and therefore provable using \irref{qear}
       \[t \in I_e \rightarrow p(t) + [-\sfrac{3\eps}{4}, \sfrac{3\eps}{4}] \subseteq I_{f(e)}\]
       As $p(t)$ has \emph{provable} error at most $\sfrac{\eps}{4}$, it further follows that the following formula is provable as well
       \[t \in I_e \rightarrow f(t) + [-\sfrac{\eps}{2}, \sfrac{\eps}{2}] \subseteq I_{f(e)}\]
       thereby completing the proof of $I$ being a provable $\sfrac{\eps}{2}$ enclosure. 
    \end{itemize}
    As the construction above handles all cases, the proof is complete. 
\end{proof}

Combining Theorems \ref{thm: compute enclosures} and \ref{thm: enclosures are provable} immediately gives the following corollary.

\begin{corollary}
    \label{cor: compute provable enclosures}
    For all $\phi \in \fmlb(\L_D)$ and domain $D : FV(\phi) \to \IQ$, for all $\eps \in \Q^+$, a \emph{provable} $\eps$-enclosure $I : \term(\phi) \to \IQ$ extending $D$ can be computed uniformly in $D, \phi, \eps$.
\end{corollary}

For the sake of simplicity in notation, this paper will only use $1$-enclosures, the results presented naturally extend to $\eps$-enclosures. Enclosures $I$ also induce corresponding domains with the identification $I \mapsto I\lvert_{FV(\phi)}$.

The following definition defines the approximations to differentially-defined functions. 

\begin{definition}[Definable approximations]
    \label{def: function approx}
    Let $\phi \in \fmlb(\L_D)$ be a bounded formula. An \emph{approximation} of $\phi$ is a function $\F : \termf(\phi) \to \folr$ such that for all $f(e) \in \termf(\phi)$, $\F_{f(e)}$ defines a univariate function $\F_{f(e)} : \R \to \R$. 
\end{definition}

Intuitively, perturbations of formulas in $\fmlb(\L_D)$ are constructed by \emph{substituting} $\folr$-definable approximations for special functions, resulting in formulas in $\folr$. The following example motivates this construction. 

\begin{example}[Perturbation of formulas]
\label{example:perturbing formulas}
    Consider the formula $\phi(x) = 2\sin(x) - 1 > 0$, with approximation $\F \coloneqq \{\sin(x) \mapsto z - \frac{z^3}{3!}\}$\footnote{Recall that $\F_{\sin(x)}$ is a ($\folr$ definition of a) function. For clarity $\F_{\sin(x)}$ is represented by its action on the dummy variable $z$.}. It is then natural to perturb $\phi$ by constructing 
    \[\tilde{\phi}(x) = 2\left(x - \frac{x^3}{3!}\right) - 1 > 0\]
    However, it is not clear how the perturbed formula $\tilde{\phi}(x)$ relates to $\phi(x)$, since $\F_{\sin(x)}$ is only an approximation. To handle this, we view $\F_{\sin(x)}(x)$ as an \emph{interval} $\F_{\sin(x)}(x) + [-\delta, \delta]$ rather than an exact value, where $\delta \in \Q^+$ is a parameter representing the accuracy of the approximation $\F_{\sin(x)}$. Perturbations of $\phi$ then amounts to substituting $\sin(x)$ with $\F_{\sin(x)}$ viewed as an interval function, yielding the following formula containing intervals depending on $\delta$.
    \[\tilde{\phi}_\delta(x) = 2\left(x - \frac{x^3}{3!} + [-\delta, \delta]\right) - 1 > 0\]
     Of course, formulas containing intervals could be interpreted semantically in different ways. This article considers the following two natural possibilities:
     \begin{itemize}
         \item For \emph{every} value $z \in [-\delta, \delta]$, the corresponding formula holds. This gives the \emph{universal perturbation} $\pertf[\F]{\phi}{\delta}$:
         \[\pertf[\F]{\phi}{\delta} = \forall^{\F_{\sin(x)}(x) + [-\delta, \delta]}w~2w - 1 > 0\]
         \item There \emph{exists} some value $z \in [-\delta, \delta]$ such that the corresponding formula holds. This gives the \emph{existential perturbation} $\perte[\F]{\phi}{\delta}$:
         \[\perte[\F]{\phi}{\delta} = \exists^{\F_{\sin(x)}(x) + [-\delta, \delta]}w~2w - 1 > 0\]
     \end{itemize}
     Note that in either case, the resulting formula has the same free-variable $x$ and belongs to $\folr$. Furthermore, for any $z \in \R$, if $\delta \in \Q^+$ satisfies $\abs{\F_{\sin(x)}(z) - \sin(z)} < \delta$, then the following chain of implications hold
     \[\pertf[\F]{\phi}{\delta}(z) \rightarrow \phi(z) \rightarrow \perte[\F]{\phi}{\delta}(z)\]
     Finally, we note that such substitutions are more delicate when special functions are composed. E.g. suppose the term $\sin(x + \cos(y))$ appears in some formula $\psi$ with an approximation given by (recall that the domain of $\G$ are function terms)
     \[\G \coloneqq \left\{\cos(y) \mapsto 1 - \frac{z^2}{2}, \sin(x + \cos(y)) \mapsto z - \frac{z^3}{3!}\right\}\]
     After first substituting in the approximation for $\cos(y)$ and thereby introducing a (bound) variable $w$, the resulting term $\sin(x + w)$ is \emph{not in the domain} of $\G$. Nonetheless, it is clear that $\sin(x + w)$ should be substituted by $\G_{\sin(x + \cos(y))}$. As such, we define a new approximation $\G_{w \mapsto \cos(y)}$ which includes $\sin(x + w)$ in its domain by ``reversing'' the substitution of $w$ back into $\cos(y)$ and evaluating the approximation given by $\G$.
     \[\G_{w \mapsto \cos(y)}(\sin(x + w)) \coloneqq \G(\sin(x + \cos(y))) = z - \frac{z^3}{3!}\]
\end{example}

Note that $\delta \in \Q^+$ is arbitrary and the construction of such perturbations are purely syntactic, semantic requirements on $\delta$ to upper-bound errors are given in Definition \ref{def: admissible approximants}.

\begin{definition}[$\delta$-perturbations]
    \label{def: perturbations}
    Let $\phi \in \fmlb(\L_D)$ be a bounded formula, $\F : \term(\phi) \to \folr$ an approximation and $\delta \in \Q^+$ a rational constant representing the level of perturbation. The \emph{universal $\delta$ perturbation of $\phi$ relative to $\F$}, denoted as $\A^\forall_\F(\phi, \delta)$ is defined inductively as follows:
    \begin{enumerate}
        \item For atomic formulas of the form $t \succeq 0$ ($\succeq~\in \rels$) where $t$ is function-free, the perturbation is the identity operator, i.e:
            \[\pert^\forall_{\F}\left(t \succeq 0, \delta\right) = t \succeq 0\]
        \item For atomic formulas of the form $t(f(e), \vec{x}) \succeq 0$ where the term $e$ is function-free the perturbation is defined as:
            \[\pert^\forall_{\F}\left(t(f(e), \vec{x}) \succeq 0, \delta\right) = \forall^{\F_{f(e)}(e) + [-\delta, \delta]} w~\pert^\forall_{\F_{w \mapsto f(e)}}\left(t(w, \vec{x}) \succeq 0, \delta\right)\]
            where $w$ is a fresh variable, $\F_{f(e)}(e) + [-\delta, \delta]$ is carried out as an addition of intervals and $\F_{w \mapsto f(e)} : \term_F(t(w, \vec{x}) \succeq 0) \to \folr$ is defined via
            \[\F_{w \mapsto f(e)}(\theta(w, \vec{x})) = \F(\theta(f(e), \vec{x}))\]
            for any function term $\theta(w, \vec{x}) \in \term_F(t(w, \vec{x}) \succeq 0)$. This carries out the construction illustrated in Example \ref{example:perturbing formulas}. 
        \item For connectives of the form $\phi_1 \lor \phi_2, \phi_1 \land \phi_2$, we have
        \[\pert^\forall_{\F}\left(\phi_1 \lor \phi_2, \delta\right) = \pert^\forall_{\F}\left(\phi_1, \delta\right) \lor \pert^\forall_{\F}\left(\phi_2, \delta\right)\]
        \[\pert^\forall_{\F}\left(\phi_1 \land \phi_2, \delta\right) = \pert^\forall_{\F}\left(\phi_1, \delta\right) \land \pert^\forall_{\F}\left(\phi_2, \delta\right)\]

        \item For sub-formulas of the form $Q^{[a, b]}y~ \psi(y, \vec{x})$, the perturbation is defined as
        \[\pert^\forall_{\F}\left(Q^{[a, b]}y~\psi(y, \vec{x}), \delta\right) = Q^{[a, b]}y~\pert^\forall_{\F}\left(\psi(y, \vec{x}), \delta\right)\]
    \end{enumerate}
    It is natural to also consider the dual notion of \emph{existential $\delta$ perturbation of $\phi$ relative to $\F$}, denoted $\pert^\exists_{\F}\left(\phi, \delta\right)$. The only difference being the following change for point (2):
    \[\pert^\exists_{\F}\left(t(f(e), \vec{x}) \succeq 0, \delta\right) = \exists^{\F_{f(e)}(e) + [-\delta, \delta]} w~\pert^\exists_{\F_{w \mapsto f(e)}}\left(t(w, \vec{x}) \succeq 0, \delta\right)\]
    Note that for both perturbations the set of free variables remain unchanged, and for $\delta_1 \geq \delta_2 > 0$ the perturbations are monotone:
    \begin{align*}
        &\pertf[\F]{\phi}{\delta_1} \rightarrow \pertf[\F]{\phi}{\delta_2}\\
        &\perte[\F]{\phi}{\delta_2} \rightarrow \perte[\F]{\phi}{\delta_1}
    \end{align*}
\end{definition}

\begin{remark}
    \label{remark: analogy for delta decidability}
    For a bounded formula $\phi \in \fmlb(\L_D)$, $\pert^\forall_{\F}\left(\phi, \delta\right)$ can also be viewed as \emph{strengthening} $\phi$ by an amount of $\delta$, and $\pert^\exists_{\F}\left(\phi(\vec{x}), \delta\right)$ as \emph{weakening} $\phi$ by an amount of $\delta$. As Theorem \ref{thm: proof rules for delta perturbations} establishes, it is indeed the case that both $\pert^\forall_{\F}\left(\phi, \delta\right) \rightarrow \phi$ and $\phi \rightarrow \pert^\exists_{\F}\left(\phi, \delta\right)$ are provably valid for appropriate choices of $\F$ (Definition \ref{def: admissible approximants}).
\end{remark}

The following result establishes how perturbations of $\phi$ and $\neg\phi$ are related. 

\begin{lemma}[Duality of $\delta$-perturbations]
    \label{lem: perturbation of negated formulas}
    Let $\phi \in \fml_B(\L_D)$ be a bounded formula, $\F : \term_F(\phi) \to \folr$ an approximation, and $\delta \in \Q^+$, then the following is provable in $\dL$:
    \[\neg\pert^\forall_{\F}\left(\phi, \delta\right) \leftrightarrow \pert^\exists_{\F}\left(\neg\phi, \delta\right)\]
\end{lemma}

\begin{proof}
    Validity of the equivalence follows directly by structural induction, and provability follows from axiom \irref{qear} since the equivalence itself is a formula in $\folr$. 
\end{proof}

The next theorem proves that the universal/existential perturbations are (provable) strengthenings/weakenings of formulas when sufficiently accurate approximations are used. 

\begin{theorem}[Provable $\delta$-Perturbations]
    \label{thm: proof rules for delta perturbations}
    Let $\phi \in \fml_B(\L_D)$ be a bounded formula, $\F : \term_F(\phi) \to \folr$ be an approximation, $0 < \delta < \frac{1}{2}$ be a rational constant and $I : \term(\phi) \to \IQ$ an enclosure of $\phi$. The following proof rules are sound\footnote{The upper-bound of $\frac{1}{2}$ is due to the use of $1$-enclosures. For general $\eps$-enclosures an upper-bound of $\frac{\eps}{2}$ suffices.}. If $I$ is furthermore a provable enclosure, then the following proof rules are derivable

    \begin{calculus}
        \cinferenceRule[univ_delta|{$\delta^\forall$}]{universal delta perturbation}
        {\linferenceRule
          {\lsequent{}{\bigwedge_{f(e) \in \term_F(\phi)} \forall^{I_{e}}s \left(\abs{\F_{f(e)}(s) - f(s)} \leq \delta\right)}}
          {\lsequent{}{\forall^{I_{\vec{x}}}\vec{x} \left(\pert^\forall_\F\left(\phi, \delta\right) \rightarrow \phi\right)}}
        }{}

        \cinferenceRule[exist_delta|{$\delta^\exists$}]{existential delta perturbation}
        {\linferenceRule
          {\lsequent{}{\bigwedge_{f(e) \in \term_F(\phi)} \forall^{I_{e}}s \left(\abs{\F_{f(e)}(s) - f(s)} \leq \delta\right)}}
          {\lsequent{}{\forall^{I_{\vec{x}}}\vec{x} \left(\phi \rightarrow \pert^\exists_\F\left(\phi, \delta\right)\right)}}
        }{}
        
    \end{calculus}
\end{theorem}

Intuitively, these proof rules show that universal/existential $\delta$-perturbations are indeed strengthenings/weakenings of the original formula, \emph{provided} that
the approximants given by $\F$ have a uniform error of at most $\delta$ on the intervals given by the enclosure $I$.

\begin{proof}
    By the soundness of $\dL$'s axiomatization, we may assume without loss of generality that $I$ is a provable enclosure. \irref{univ_delta} is first derived via simultaneous induction on the construction of $\pert^\forall_\F\left(\phi, \delta\right)$ for all $\phi, I$. The base case of function-free atomic formulas is trivial since the perturbation is the identity operator on such formulas, and similarly the case for connectives/quantifiers follow directly by inductive hypothesis. It suffices to just consider the non-trivial case $\phi(\vec{x}) = t(f(e), \vec{x}) \succeq 0$ where $e$ is function-free. Recall by construction of $\pert^\forall_\F(\phi(x), \delta)$ we have
        \[\pert^\forall_{\F}\left(t(f(e), \vec{x}) \succeq 0, \delta\right) = \forall^{[-\delta, \delta] + \F_{f(e)}(e)} w \left( \pert^\forall_{\F_{w \mapsto f(e)}}\left(t(w, \vec{x}) \succeq 0, \delta\right)\right)\]
        Standard reductions then give the following derivation:
        \begin{sequentdeduction}
            \linfer[]
                {\linfer[cut+alll]
                    {\linfer[]
                        {\text{\textcircled{1}}}
                    {\lsequent{}{\forall^{I_{\vec{x}}}\vec{x}\left(\abs{\F_{f(e)}(e) - f(e)} \leq \delta\right)}}
                    &
                    {\text{\textcircled{2}}}
                    }
                {\lsequent{}{\forall^{I_{\vec{x}}}\vec{x}\left(\forall^{[-\delta, \delta] + \F_{f(e)}(e)} w \left( \pert^\forall_{\F_{w \mapsto f(e)}}\left(t(w, \vec{x}) \succeq 0, \delta\right)\right) \rightarrow t(f(e), \vec{x}) \succeq 0\right)}}
                }
            {\lsequent{}{\forall^{I_{\vec{x}}}\vec{x}\left(\pert^\forall_\F\left(\phi(\vec{x}), \delta\right) \rightarrow \phi(\vec{x})\right)}}
        \end{sequentdeduction}
        Where 
        \[\text{\textcircled{2}} \equiv \forall^{I_{\vec{x}}}\vec{x}\forall^{[-\delta, \delta] + \F_{f(e)}(e)} w \left(\pert^\forall_{\F_{w \mapsto f(e)}}\left(t(w, \vec{x}) \succeq 0, \delta\right) \rightarrow t(w, \vec{x}) \succeq 0\right)\]
        To handle premise \textcircled{1}, recall the enclosure $I$ was assumed to be provable. In particular, this means that the following is provable in $\dL$:
        \[\forall^{I_{\vec{x}}}\vec{x} \left(e(\vec{x}) \in I_{e}\right)\]
        Thus, \textcircled{1} can be handled as follows:
        \begin{sequentdeduction}
            \linfer[cut]
                {\linfer[]
                    {\lclose}
                {\lsequent{}{\forall^{I_{\vec{x}}}\vec{x} \left(e(\vec{x}) \in I_{e}\right)}}
                &
                \linfer[]{}
                {\lsequent{}{\forall^{I_{e}}s \left(\abs{\F_{f(e)}(s) - f(s)} \leq \delta\right)}}
                }
            {\lsequent{}{\forall^{I_{\vec{x}}}\vec{x}\left(\abs{\F_{f(e)}(e) - f(e)} \leq \delta\right)}}
        \end{sequentdeduction}
        which is of the desired form since $f(e) \in \term_F(\phi)$. It remains to derive premise \textcircled{2}. First note that $\delta < \frac{1}{2}$ implies the provability of the following
        \[ \lsequent{\forall^{I_{e}}s \left(\abs{\F_{f(e)}(s) - f(s)} \leq \delta\right)}{\forall^{I_{\vec{x}}}\vec{x}\forall^{[-\delta, \delta] + \F_{f(e)}(e)}w \abs{w - f(e)} \leq 1}\]
        Hence, standard deductions and the provability of $I$ further imply that following is provable as well
        \[\lsequent{\forall^{I_{e}}s \left(\abs{\F_{f(e)}(s) - f(s)} \leq \delta\right)}{\forall^{I_{\vec{x}}}\vec{x}\forall^{[-\delta, \delta] + \F_{f(e)}(e)}w \left(w \in I_{f(e)}\right)}\]
        Chaining these up results in the following derivation of \textcircled{2}:
        \begin{sequentdeduction}
            \linfer[cut]
                {\linfer[]
                    {}
                {\lsequent{}{\forall^{I_{e}}s \left(\abs{\F_{f(e)}(s) - f(s)} \leq \delta\right)}}
                &
                \linfer[]
                    {\text{\textcircled{3}}}
                {\lsequent{}{\forall^{I_{\vec{x}}}\vec{x}\forall^{I_{f(e)}}w \left(\pert^\forall_{\F_{w \mapsto f(e)}}\left(t(w, \vec{x}) \succeq 0, \delta\right) \rightarrow t(w, \vec{x}) \succeq 0\right)}}
                }
            {\lsequent{}{\forall^{I_{\vec{x}}}\vec{x}\forall^{[-\delta, \delta] + \F_{f(e)}(e)} w \left(\pert^\forall_{\F_{w \mapsto f(e)}}\left(t(w, \vec{x}) \succeq 0, \delta\right) \rightarrow t(w, \vec{x}) \succeq 0\right)}}
        \end{sequentdeduction}
        Notice that the left premise is of the form we want, so it suffices to prove premise \textcircled{3}. To this end, define the enclosure $J : \term(t(w, \vec{x}) \succeq 0) \to \IQ$ for the formula $t(w, \vec{x}) \succeq 0$ via the following for all terms $s(w, \vec{x}) \in \term(t(w, \vec{x}) \succeq 0)$:
        \[J_{s(w, \vec{x})} = I_{s(f(e), \vec{x})}\]
        This is well-defined sine $w$ is the only new variable introduced, it is indeed an (provable) enclosure since $I$ is an (provable) enclosure. Therefore, we may invoke our inductive hypothesis on $t(w, \vec{x}) \succeq 0$ and $J$. Applying this to premise \textcircled{3} and simplifying yields the following proof (where the use of inductive hypothesis is indicated with \irref{univ_delta}):
        
        \begin{sequentdeduction}
            \linfer[univ_delta]
                {\linfer[id]
                    {\linfer[id]
                        {\lsequent{}{\bigwedge_{f(e) \in \term_F(\phi)} \forall^{I_{e}}s \left(\abs{\F_{f(e)}(s) - f(s)} \leq \delta\right)}}
                    {\lsequent{}{\bigwedge_{g(u(f(e), \vec{x})) \in \term_F(\phi)} \forall^{I_{u(f(e), \vec{x})}}s \left(\abs{\F_{g(u(f(e), \vec{x}))}(s) - g(s)} \leq \delta\right)}}
                    }
                {\lsequent{}{\bigwedge_{g(u(w, \vec{x})) \in \term_F(t(w, \vec{x}) \succeq 0)} \forall^{J_{u(w, \vec{x})}}s \left(\abs{\F_{w \mapsto f(e)}(g(u(w, \vec{x})))(s) - g(s)} \leq \delta\right)}}    
                }
            {\lsequent{}{\forall^{I_{\vec{x}}}\vec{x}\forall^{I_{f(e)}}w \left(\pert^\forall_{\F_{w \mapsto f(e)}}\left(t(w, \vec{x}) \succeq 0, \delta\right) \rightarrow t(w, \vec{x}) \succeq 0\right)}}
        \end{sequentdeduction}
    Since the remaining premise is precisely the premise of \irref{univ_delta}, this completes the proof.

    This completes the derivation of \irref{univ_delta}, and the derivation of \irref{exist_delta} now directly follows via duality, noticing that function terms and enclosures are preserved by taking negations. We have:
    \begin{sequentdeduction}
        \linfer[]
            {\linfer[]
                {\linfer[univ_delta]
                    {\lsequent{}{\bigwedge_{f(e) \in \term_F(\phi(\vec{x}))} \forall^{I_{e}}s \left(\abs{\F_{f(e)}(s) - f(s)} \leq \delta\right)}}
                {\lsequent{}{\forall^{I_{\vec{x}}}\vec{x} \left(\pert^\forall_\F\left(\neg\phi(\vec{x}), \delta\right) \rightarrow \neg\phi(\vec{x})\right)}}
                }
            {\lsequent{}{\forall^{I_{\vec{x}}}\vec{x} \left(\neg\pert^\exists_\F\left(\phi(\vec{x}), \delta\right) \rightarrow \neg\phi(\vec{x})\right)}}
            }
        {\lsequent{}{\forall^{I_{\vec{x}}}\vec{x} \left(\phi(\vec{x}) \rightarrow \pert^\exists_\F\left(\phi(\vec{x}), \delta\right)\right)}}
    \end{sequentdeduction}
    This completes the proof of Theorem \ref{thm: proof rules for delta perturbations}.
\end{proof}

Theorem \ref{thm: proof rules for delta perturbations} establishes that the perturbations $\pertf{\phi}{\delta}, \perte{\phi}{\delta}$ are \emph{provable} $\delta$-strengthenings (weakenings) by uniformly reducing all numerical conditions down to requirements on the accuracy of $\F$. Furthermore, since such perturbations are $\folr$ formulas, they admit quantifier elimination. As a corollary, it follows that formulas $\phi \in \fmlb(\L_D)$ can be \emph{provably} approximated by quantifier-free $\folr$ formulas. The following example illustrates this.

\begin{example}[Symbolic Conditions for Lyapunov Functions \footnote{This example has been formally verified using the $\dL$ theorem prover KeYmaera X \cite{DBLP:conf/cade/FultonMQVP15} with exact arithmetic.}]
    \label{example: Lyapunov function approx}
    Let $\phi(c_1, c_2, c_3)$ be the $\Pi_1$-formula from Example \ref{example: lyapunov function} which is true if and only if the function $V(x_1, x_2)$ is a valid (strict) Lyapunov function. Earlier works \cite{DBLP:conf/cav/KongSG18, DBLP:conf/hybrid/KapinskiDSA14} synthesized \emph{specific} numerical values of $\vec{c} = (c_1, c_2, c_3)$ such that $\phi(c_1, c_2, c_3)$ is satisfied when $\vec{x}$ is bounded away from the origin. I.e. The quantification on $\vec{x}$ is relaxed to the set $\{\eps < \norm{\vec{x}} < 1\}$ for some positive $\eps > 0$. By utilizing Theorem \ref{thm: proof rules for delta perturbations}, one can further synthesize \emph{logical characterizations} of $\vec{c}$ that provably satisfies the formula $\phi(\vec{c})$ (bounded away from $0$). Concretely, we modify $\phi(\vec{c})$ to quantify over $\{0.1 \leq \norm{\vec{x}}_{\infty} \leq 0.5\}$ and denote this relaxed formula by $\tilde{\phi}$. This is similar to earlier work \cite[Section~5.2]{DBLP:conf/cav/KongSG18} which considered the set $\{0.1 \leq \norm{\vec{x}} \leq 1\}$. As Lyapunov functions are only required to hold for some open ball around the origin, the upperbound on $\norm{\vec{x}}$ can be chosen freely. $l_\infty$ norm was used to better handle the restriction on quantifiers. 
    
    Consider the perturbation $\pertf[\F]{\tilde{\phi}}{0.0052}$ given by $\F = \left\{\sin(x_1) \mapsto \frac{969}{1000}z\right\}$, resulting in the $\folr$ formula 
    \[\pertf[\F]{\tilde{\phi}}{0.0052} \equiv \forall^{0.1 \leq \norm{\vec{x}}_\infty \leq 0.5}\vec{x} \forall^{\F_{\sin(x_1)}(x_1) + [-0.0052, 0.0052]} w~\psi(\vec{c}, x_1, x_2, w)\]
    where 
    \[\psi(\vec{c}, x_1, x_2, w) \equiv V(x, y) > 0 \land (c_1x_1 + 2c_3x_2)(x_2 + w) -x_2(2c_2x_1 + c_1x_2) > 0\]
    is the quantifier-free part. Further let $I$ denote the enclosure of $\tilde{\phi}$ given by $I_{x_1} = I_{x_2} = [-0.5, 0.5], I_{c_1} = I_{c_2} = I_{c_3} = [0, 100]$ and $I_{\sin(x_1)} = [-2, 2]$ (note that any $I_{\sin(x_1)}$ suffices as the term $\sin(x_1)$ is not used in other function terms), it follows (e.g. by numerical computations) that 
    \[\norm{\F_{\sin(x_1)}(x_1) - \sin(x_1)}_{[-0.5, 0.5]} < 0.0052\]
    Theorem \ref{thm: stone weierstrass for IVPs} then implies that this is \emph{provable}, and therefore an application of the proof-rule \irref{univ_delta} \emph{syntactically proves} 
    \[\forall^{[0, 100]^3}\vec{c}~(\pertf[\F]{\tilde{\phi}}{0.0052} \rightarrow \tilde{\phi})\]
    In particular, one can apply quantifier elimination on $\pertf[\F]{\phi}{0.0052}$ to obtain a provably correct under-approximation of $\eval{\phi}$ that only involves the variables $c_1, c_2, c_3$. By a slight abuse of notation, let $\Psi(c_1, c_2, c_3)$ denote the formula obtained after applying quantifier elimination, this then easily gives \emph{provable symbolic constrains} on $c_1, c_2, c_3$ that guarantees $V(x_1, x_2)$ to be a valid Lyapunov function on $\{0.1 \leq \norm{\vec{x}}_\infty \leq 0.5\}$. For example, simplifying $\Psi(40.6843, 35.6870, c_3)$ gives
    \[\Psi(40.6843, 35.6870, c_3) \equiv \frac{ 274856791 -\sqrt{1148902878104790}}{4845000} < c_3 < \frac{56029150 + \sqrt{1213829504113710}}{1021000}\]
    and recovers the tuple $\vec{c} = (40.6843, 35.6870, 84.3906)$ found in earlier work \cite[Section~5.2]{DBLP:conf/cav/KongSG18} as a special case. For better presentability, this can be rounded conservatively to obtain
    \[22 < c_3 < 88 \rightarrow \Psi(40.6843, 35.6870, c_3) \implies 22 < c_3 < 88 \rightarrow \tilde{\phi}(40.6843, 35.6870, c_3)\]
    Thus, this example shows how such perturbations can be used to generate \emph{provable} under-approximations that are symbolic in the variables, beyond the capabilities of approximate decision procedures \cite{Gao_Kong_Chen_Clarke_2014,DBLP:conf/cav/KongSG18,Franek_Ratschan_Zgliczynski_2016} which can only handle sentences without free variables. 
\end{example}

Motivated by the premises appearing in Theorem \ref{thm: proof rules for delta perturbations}, one can define a natural norm on approximations to formulas, leading to the notion of \emph{admissible approximations}.

\begin{definition}[Admissible approximations]
    \label{def: admissible approximants}
    Let $\phi \in \fml_B(\L_D)$ be a bounded formula, $I : \term(\phi) \to \IQ$ an enclosure and $\F : \term_F(\phi) \to \folr$ an approximation. The norm of $\F$ relative to enclosure $I$, denoted $\norm{\F}_I$, is defined as:
    \[ \norm{\F}_I = \max_{f(e) \in \termf(\phi)} \norm{\F_{f(e)} - f}_{I_e}\]
    For a given $\delta \in \Q^+$ such that $\delta < \frac{1}{2}$, the approximation $\F$ is said to be: 
    \begin{itemize}
        \item \emph{$(I, \delta)$ admissible} if $\norm{\F}_I < \delta$.
        \item \emph{provably} $(I, \delta)$ admissible if $\bigwedge_{f(e) \in \termf(\phi)} \norm{\F_{f(e)} - f}_{I_e} < \delta$ is provable (in $\dL$).
    \end{itemize}
    A sequence of approximants $(\F_{\delta_i})_i$ with $\delta_i \rightarrow 0$ is said to be (provably) \emph{$I$ admissible} if for all $i \in \N$, $\F_{\delta_i}$ is (provably) $(I, \delta_i)$ admissible. For brevity, such sequences are indexed as $(\F_\delta)_\delta$ with the understanding that $\delta \rightarrow 0$ and $\delta_i < \frac{1}{2}$. 
\end{definition}

\begin{remark}
    The notion of admissibility introduced in Definition \ref{def: admissible approximants} requires $\norm{\F}_I < \delta$, which is slightly stronger than the premise of $\norm{\F}_I \leq \delta$ appearing in Theorem \ref{thm: proof rules for delta perturbations}. This difference arises to allow for a more uniform treatment of robustness in Section \ref{sec: robustness}. 
\end{remark}

The existence of provable $(I, \delta)$ admissible approximations for all $\delta \in \Q^+$ follows directly from Theorem \ref{thm: stone weierstrass for functions}.

\begin{theorem}
    \label{thm: provable approximations exist}
    Let $\phi \in \fmlb(\L_D)$ be a bounded formula, $I : \term(\phi) \to \IQ$ an enclosure and $\delta \in \Q^+$. Then there (computably) exists a provable $(I, \delta)$ approximation $\F : \termf(\phi) \to \folr$.
\end{theorem}

\begin{proof}
    Let $f(e) \in \termf(\phi)$ be an arbitrary function term, let $T \in \Q^+$ be large enough that $I_e \subseteq [-T, T]$. Let $p \in \Q[t]$ be some polynomial approximation of $f$ on $[-T, T]$ with provable error at most $\delta$ as computed by Theorem \ref{thm: stone weierstrass for functions}, it then follows that $\F_{f(e)} \coloneqq p$ constructs a provable $(I, \delta)$ admissible approximation. 
\end{proof}

Consequently, it follows that $\dL$'s axiomatization is $\delta$-complete.

\begin{theorem}[$\delta$-Completeness]
    \label{thm: delta completeness}
    For every bounded sentence $\phi \in \sentb(\L_D)$, every $\delta \in \Q^+$, there are corresponding (computable) syntactic approximations $\pertf{\phi}{\delta}, \perte{\phi}{\delta}$ such that the following hold:
    \begin{itemize}
        \item If $\phi$ is true, then the approximation $\perte{\phi}{\delta}$ is provable
        \begin{align*}
            \R_D \models \phi \implies& \vdash_{\dL} \perte{\phi}{\delta}\\
            \intertext{\item If $\phi$ is false, then the negation of the approximation $\pertf{\phi}{\delta}$ is provable}
            \R_D \models \neg \phi \implies& \vdash_{\dL} \neg\pertf{\phi}{\delta}
        \end{align*}
        \item Furthermore, the following are provable.
        \begin{align*}
            \vdash_{\dL} &~\phi \rightarrow \perte{\phi}{\delta}\\
            \vdash_{\dL} &~\pertf{\phi}{\delta} \rightarrow \phi
        \end{align*}
    \end{itemize}
\end{theorem}

\begin{proof}
    Let $(\phi, \delta) \in \sentb(\L_D) \times \Q^+$ be given. Without loss generality, we may assume $\delta < \frac{1}{2}$ (one can always take $\delta = \min\left(\delta, \frac{1}{4}\right)$). Let $I : \term(\phi) \to \IQ$ be a provable enclosure of $\phi$ as computed by Theorem \ref{cor: compute provable enclosures}, and $\F : \term_F(\phi) \to \folr$ some provable $(I, \delta)$ admissible approximation which is guaranteed to (computably) exist by Theorem \ref{thm: provable approximations exist}. The desired approximations $\pertf{\phi}{\delta}, \perte{\phi}{\delta}$ can now be computably constructed as $\pertf[\F]{\phi}{\delta}, \perte[\F]{\phi}{\delta}$, and provability follows directly by Theorem \ref{thm: proof rules for delta perturbations} and the fact that $\F$ is a provably admissible approximation. 
\end{proof}

\begin{theorem}[Proof-producing $\delta$-decidability]
    \label{thm: delta decidability with proofs}
    There is a computable algorithm that takes in bounded sentences $\phi \in \sentb(\L_D)$, rational constants $\delta > 0$ and outputs the following:
    \begin{enumerate}
        \item A bounded sentence $\pert^\forall(\phi, \delta) \in \sentb(\L_D)$, the ``$\delta$ strengthening'' of $\phi$.
        
        \item The corresponding truth value of $\pert^\forall(\phi, \delta)$.
        
        \item In the case that $\pert^\forall(\phi, \delta)$ is true, a proof of $\phi$ being true in $\dL$.  
    \end{enumerate}
\end{theorem}

\begin{remark}
    Theorem \ref{thm: delta decidability with proofs} is precisely what this article means by ``$\delta$-decidability'' with proofs. Output (1) of the algorithm is a purely syntactic modification of the sentence, output (2) decides the truth value of the perturbed $\folr$ sentence, for which there are existing proof-generating decision procedures \cite{DBLP:conf/cade/McLaughlinH05}. In the case that the $\delta$ strengthening of $\phi$ is true, output (3) further produces a \emph{symbolic proof} in $\dL$ of the truth of $\phi$.  
\end{remark}

\begin{proof}
     The algorithm can simply output the desired strengthening $\pertf{\phi}{\delta}$ as given by Theorem \ref{thm: delta completeness}. For output (2), the algorithm outputs the truth value of $\pert^\forall(\phi, \delta)$, which is computable as $\folr$ is decidable. Thus, it remains to show (3) and output a syntactic proof of $\phi$ in the event that $\pert^\forall(\phi, \delta)$ is true. Indeed, applying the proof rule \irref{univ_delta} yields:

    {\renewcommand{\linferPremissSeparation}{\quad}
    \begin{sequentdeduction}[array]
        \linfer[cut]
        {
            \linfer[qear]{\lclose}
            {\lsequent{}{\pertf{\phi}{\delta}}}
        !
         \linfer[univ_delta]
            {\lclose}
        {\lsequent{}{\pertf{\phi}{\delta} \rightarrow \phi}}
        }
        {\lsequent{}{\phi}}
    \end{sequentdeduction}
    }
    
    Where the application of \irref{univ_delta} is sound per the construction of $\pertf{\phi}{\delta}$ via Theorem \ref{thm: delta completeness}. Consequently, the derivation above then serves as a syntactic proof for the truth $\phi$, since the $\folr$ sentence $\pertf{\phi}{\delta}$ can always be proven by \irref{qear} when it is true. This completes the computation of output (3), and therefore the proof of the theorem. 
\end{proof}

\begin{remark}
    The proof of Theorem \ref{thm: delta decidability with proofs} proves something stronger, that \emph{every} $(I, \delta)$ admissible approximation $\F$ gives rise to a $\delta$ decision procedure. For future works it would be interesting to determine the most useful approximations for domain-specific problems. 
\end{remark}

\section{Robustness of perturbations}
\label{sec: robustness}
This section introduces the notion of \emph{robustness} for formulas relative to the perturbations defined in Section \ref{sec: perturbations}. Intuitively, robust formulas are those formulas whose truth values are continuous under perturbations. If the level of perturbation $\delta \to 0$, then the truth values of the perturbed formulas eventually agree with the original formula. This section establishes that those formulas whose prenex normal form only contains one of the two relation symbols $\rels$ are robust. It is also shown that there are inherent computability-theoretic obstacles in extending such results to all of $\R_C$. We first begin by defining the set of pure formulas.

\begin{definition}[Pure formulas]
    \label{def: pure formulas}
    A formula $\phi \in \fml(\L_C)$ is said to be \emph{pure} if its normal form only contains one type of relation symbol $\succeq~\in \rels$, in which case it is also said to be $\succeq$-pure. 
\end{definition}

The following definitions will be useful in handling bounded formulas, recall that the domain of a bounded formula $\phi$ is a function $D : FV(\phi) \to \IQ$.

\begin{definition}[Bounded semantics]
    \label{def: bounded semantics}
    Let $\phi(\vec{x}) \in \fmlb(\L_D)$ be a bounded formula of arity $n$ ($\abs{FV(\phi)} = n$) and $D : FV(\phi) \to \IQ$ a domain of $\phi$. The \emph{bounded semantics} of $\phi$ relative to $D$, $\evalb{\phi(\vec{x})}$, is defined as
    \[\evalb{\phi(\vec{x})} = D \cap \eval{\phi(\vec{x})} = \{\vec{y} \in \eval{\phi(\vec{x})}~\vert~\forall 1 \leq i \leq n(y_i \in D_{x_i})\}\]
    Correspondingly, write $D \models \phi(\vec{x})$ if and only if $\evalb{\phi(\vec{x})} = D$. If $\phi$ is a sentence and $FV(\phi) = \emptyset$, we define $\evalb{\phi} = \eval{\phi} \in \{\top, \bot\}$ to be the truth value of $\phi$.
\end{definition}

The robustness of a bounded formula $\phi$ can now be defined relative to a given domain. 

\begin{definition}[Robustness]
    \label{def: robustness}
    Let $\phi \in \fml(\L_F)$ be a bounded formula, and $D : FV(\phi) \to \IQ$ a domain of $\phi$, $\phi$ is said to be:
    \begin{enumerate}
        \item \emph{$\forall$-robust relative to $D$}, if there exists an enclosure $I : \term(\phi) \to \IQ$ extending $D$ and a sequence of $I$ admissible approximations $(\F_\delta)_\delta : \termf(\phi) \to \folr$ such that
        \[\evalb{\phi} = \bigcup_{\delta > 0} \evalb{\pertf[\F_\delta]{\phi}{\delta}}\]
       \item \emph{$\exists$-robust relative to $D$}, if there exists an enclosure $I : \term(\phi) \to \IQ$ extending $D$ and a sequence of $I$ admissible approximations $(\F_\delta)_\delta : \termf(\phi) \to \folr$ such that
        \[\evalb{\phi} = \bigcap_{\delta > 0} \evalb{\perte[\F_\delta]{\phi}{\delta}}\]
    \end{enumerate}
    If $\phi$ is a sentence, $\evalb{\pertf[\F_\delta]{\phi}{\delta}}$ (and likewise for $\perte[\F_\delta]{\phi}{\delta}$) is a monotone sequence in $\{\bot, \top\}$, hence the operators $\bigcup_{\delta > 0}, \bigcap_{\delta > 0}$ can be replaced by $\lim_{\delta \to 0}$. In the case that a formula $\psi$ is both $\forall$ and $\exists$-robust relative to $D$, $\psi$ is said to be \emph{strongly robust relative to $D$}. 
\end{definition}

For brevity, this article refers to $\forall$-robustness just by robustness. All results concerning $\forall$-robustness admit natural corresponding results for $\exists$-robustness using Lemma \ref{lem: perturbation of negated formulas}. As Theorem \ref{thm: atomic formulas are robust} will establish, bounded $>$-pure formulas are always robust for all domains. 

\begin{remark}
    By definition, the set of robust sentences (note that this does not depend on a choice of domain since sentences have no free variables) is 
    \[\{\phi \in \sentb(\L_D)~\vert~\text{$\phi$ is $\forall$-robust}\}\]
    It is natural to consider a different notion of robustness where a sentence $\phi$ is said to be robust if it satisfies at least one of the requirements in Definition \ref{def: robustness}, the set of robust sentences would then be
    \[\{\phi \in \sentb(\L_D)~\vert~\text{$\phi$ is $\forall$-robust or $\exists$-robust}\}\]
    However since every true sentence is trivially $\exists$-robust and every false sentence is trivially $\forall$-robust, this would imply that all sentences are either $\exists$-robust or $\forall$-robust, trivializing the definition. 
\end{remark}

While the robustness of a formula in Definition \ref{def: robustness} seems to be dependent on the enclosure and approximations chosen, the following result shows that is is not the case and robustness is an intrinsic property of the formula in question.  

\begin{theorem}[Characterization of robustness]
    \label{thm: chacterization of robustness}
    Let $\phi \in \fmlb(\L_D)$ be a bounded formula and $D : FV(\phi) \to \IQ$ a domain of $\phi$. The following are equivalent:
    \begin{enumerate}
        \item $\phi$ is $\forall$-robust relative to $D$.
        \item For \emph{all} enclosures $I : \term(\phi) \to \IQ$ of $\phi$ extending $D$, for \emph{all} sequences of $I$ admissible approximations $(\F_\delta)_\delta$, the following holds:
        \[\evalb{\phi} = \bigcup_{\delta > 0} \evalb{\pertf[\F_\delta]{\phi}{\delta}}\]
    \end{enumerate}
    Dual equivalences hold for $\exists$-robustness.
    \begin{enumerate}
        \item $\phi$ is $\exists$-robust relative to $D$.
        \item For \emph{all} enclosures $I : \term(\phi) \to \IQ$ of $\phi$ extending $D$, for \emph{all} sequences of $I$ admissible approximations $(\F_\delta)_\delta$, the following holds:
        \[\evalb{\phi} = \bigcap_{\delta > 0} \evalb{\pertf[\F_\delta]{\phi}{\delta}}\]
    \end{enumerate}
\end{theorem}

\begin{remark}[Robust$\iff$Continuous truth]
    An alternate interpretation of Theorem \ref{thm: chacterization of robustness} is that robust sentences are \emph{exactly} those with ``continuous truth values'' in the following sense: Let $T : \sentb(\L_D) \to \{0, 1\}$ be the evaluation function defined by:
    \begin{align*}
        &T(\phi) = 1 \iff \phi \in \theoryb(\R_D)\\
        &T(\phi) = 0 \iff \phi \notin \theoryb(\R_D)
    \end{align*}
    Then a sentence $\phi \in \sentb(\L_D)$ is $\forall$-robust if and only if 
    \[\lim_{\delta \to 0} T(\pertf[\F_\delta]{\phi}{\delta}) = T(\phi)\]
    for all $I$ admissible approximations $(\F_\delta)_\delta$. Similarly, $\phi$ is $\exists$-robust if and only if
    \[\lim_{\delta \to 0} T(\perte[\F_\delta]{\phi}{\delta}) = T(\phi)\]
    thus $\phi$ is strongly robust if and only if
    \[\lim_{\delta \to 0} T(\perte[\F_\delta]{\phi}{\delta}) = \lim_{\delta \to 0} T(\pertf[\F_\delta]{\phi}{\delta}) = T(\phi)\]
\end{remark}

The following lemma establishing a relationship between perturbations induced by different approximations will be useful in proving Theorem \ref{thm: chacterization of robustness}. 

\begin{lemma}[Transfer of approximations]
    \label{lem: mont perturbations}
    Let $\phi \in \fmlb(\L_D)$ be a bounded formula, $I : \term(\phi) \to \IQ$ an enclosure, $0 < \delta < \frac{1}{2}$ and $\F: \termf(\phi) \to \folr$ an $(I, \delta)$ admissible approximation. Then for all $(I, \frac{\delta - \norm{F}_I}{2})$ admissible approximations $\G : \termf(\phi) \to \folr$, the following hold ($I$ is identified with the domain $I\lvert_{FV(\phi)}$):
    \[I \models \pertf[\F]{\phi}{\delta} \rightarrow \pertf[\G]{\phi}{\frac{\delta - \norm{\F}_I}{2}}\]
\end{lemma}

\begin{proof}
    See Appendix \ref{app: proofs}.
\end{proof}

With Lemma \ref{lem: mont perturbations}, Theorem \ref{thm: chacterization of robustness} can now be proven. 

\begin{proof}[Proof of Theorem \ref{thm: chacterization of robustness}]
    $(1)\!\implies\!(2)$: As $\phi$ is $\forall$-robust relative to $D$, let $I : \term(\phi) \to \IQ$ be an enclosure extending $D$ and $(\F_\delta)_\delta : \termf(\phi) \to \folr$ a sequence of $I$ admissible approximations such that 
    \[\evalb{\phi} = \bigcup_{\delta > 0} \evalb{\pertf[\F_\delta]{\phi}{\delta}}\]
    For an arbitrary enclosure $J : \term(\phi) \to \IQ$ extending $D$ and an arbitrary sequence of $J$ admissible approximations $(\G_\eps)_\eps$, we need to establish the following (the converse is trivial as $(\G_\eps)_\eps$ is admissible):
    \[\evalb{\phi} \subseteq \bigcup_{\eps > 0} \evalb{\pertf[\G_\eps]{\phi}{\eps}}\]
    let $I \cap J$ denote the enclosure defined by $(I \cap J)_e \coloneqq I_e \cap J_e$. Since $(\F_\delta)_\delta$ is $I$ admissible, it is furthermore $I \cap J$ admissible, and the same applies to $(\G_\eps)_\eps$. Lemma \ref{lem: mont perturbations} shows that for all $\delta$, for all sufficiently small $\eps > 0$ satisfying $\eps < \frac{\delta - \norm{\F}_{I \cap J}}{2}$, the following hold
    \[\evalb{\pertf[\F_\delta]{\phi}{\delta}} \subseteq \evalb{\pertf[\G_\eps]{\phi}{\eps}}\]
    Consequently, this yields the desired:
    \[\evalb{\phi} = \bigcup_{\delta} \evalb{\pertf[\F_\delta]{\phi}{\delta}} \subseteq \bigcup_{\eps} \evalb{\pertf[\G_\eps]{\phi}{\eps}}\]
    $(2) \implies (1)$: This follows directly by definition of robustness.\\
    Corresponding equivalences for $\exists$-robustness can be obtained by taking the contrapositive and using Lemma \ref{lem: perturbation of negated formulas}. 
\end{proof}

Unfortunately but also unsurprisingly, not all formulas in $\fmlb(\L_D)$ are robust.

\begin{example}[Non-robust sentence]
    \label{ex: non-robust formula}
    Consider the true sentence 
    \[\phi \equiv \forall^{[0, 1]}x \left(e^x - 1 > 0 \lor \left(1 - e^x \geq 0 \land e^x - 1 \geq 0\right)\right)\]
    note that $e^x$ is indeed differentially definable via the IVP $y' = y, y(0) = 1$ with provable global existence since its corresponding ODE is linear \cite{DBLP:journals/fac/TanP21}. For all $0 < \delta < \frac{1}{2}$, let $p_\delta(x)$ be a polynomial satisfying $\norm{p_\delta(x) - e^x}_{[0, 1]} < \delta$ and $p_\delta(0) = 1$ (e.g. Taylor expansion of sufficiently high degree) so that the sequence of approximations $(e^x \mapsto p_\delta(x))_\delta$ is $I$ admissible where $I$ is the enclosure of $\phi$ defined by $I_x = [0, 1], I_{e^x} = [0, 3]$ (since $e^x$ is not the argument of any other function term, the value of $I_{e^x}$ can be assigned arbitrarily and has no effect). For any $\delta$, the corresponding perturbation is then
    \[\pert^\forall_{e^x \mapsto p_\delta(x)}(\phi, \delta) \equiv \forall^{[0, 1]}x \forall w \left(\abs{w - p_\delta(x)} \leq \delta \rightarrow w - 1 > 0 \lor \left(1 - w \geq 0 \land w - 1 \geq 0\right)\right)\]
    which is not satisfied at $x = 0, w = 1 - \frac{\delta}{2}$, and therefore the perturbed sentence is false. Since this construction works for all $\delta > 0$, Theorem \ref{thm: chacterization of robustness} shows that $\phi$ is not $\forall$-robust. Furthermore, Lemma \ref{lem: perturbation of negated formulas} implies that $\neg \phi$ will not be $\exists$-robust. 
\end{example}

While Example \ref{ex: non-robust formula} shows that not all sentences are robust, the following theorem shows that robustness holds for a wide syntactic class of formulas. 

\begin{theorem}[Pure formulas are robust]
    \label{thm: atomic formulas are robust}
    Let $\succeq~\in \rels$ be a relation symbol and $\phi \in \fml_B^\succeq(\L_D)$ a bounded $\succeq$-pure formula, then:
    \begin{itemize}
        \item If $\succeq = >$, then $\phi$ is $\forall$-robust relative to all domains.
        \item If $\succeq = \geq$, then $\phi$ is $\exists$-robust relative to all domains. 
    \end{itemize}
\end{theorem}

Such a syntactic restriction can also be viewed topologically, namely that every formula $\phi \in \fmlo_B(\L_D)$ characterizes an \emph{open set}, and every formula $\phi \in \fmlc_B(\L_D)$ characterizes a \emph{closed set}, the proof of Theorem \ref{thm: atomic formulas are robust} crucially relies on such topological properties. Note that this importantly requires the formulas to have \emph{bounded quantifiers}. For example, the formula $\phi(x, y) = xy - 1 \geq 0$ defines a closed subset of $\R^2$, yet the following formula with unbounded quantification defines the open set $\R \setminus \{0\}$.  
\[\psi(x) = \exists y~xy - 1 \geq 0\]
Intuitively, the preservation of openness/closedness amounts to requiring the projection map $\proj : X \times Y \to X$ to be open/closed, which only holds if $Y$ is compact. The following lemma gives an alternate interpretation of this fact. 

\begin{lemma}
    \label{lem: continuously choose interior}
    Let $\phi \in \fmlo_B(\L_D)$ be a bounded formula of arity $n$. There exists a continuous function $r_\phi : \R^n \to \R^{\geq 0}$ such that the following are equivalent for all $\vec{z} \in \R^n$:
    \begin{enumerate}
        \item $\vec{z} \in \eval{\phi}$
        \item $r_\phi(\vec{z}) > 0$
        \item $r_\phi(\vec{z}) > 0$ and $B(\vec{z}, r_\phi(\vec{z})) \subseteq \eval{\phi}$, where $B(\vec{z}, r_\phi(\vec{z}))$ denotes the open ball of radius $r_\phi(\vec{z})$ centered at $\vec{z}$. 
    \end{enumerate}
\end{lemma}

\begin{proof}
    It suffices to show that $\eval{\phi}$ is an open set, since then $r_\phi$ can be taken to be the distance function to $\eval{\neg \phi}$ and the claim follows. First note that one can find some bounded formula $\psi \equiv \phi$ such that $\psi$ only quantifies over $[0, 1]$, and it suffices to prove the claim for $\psi$ as $\eval{\phi} = \eval{\psi}$. We will prove the claim for $\psi$ by induction, noting that the base case for atomic formulas follow directly by continuity of terms. Suppose $\psi(x) \equiv\forall^{[0, 1]}y~\eta(x, y)$, giving
    \[\eval{\psi} = \{x \in \R~\vert~\forall^{[0, 1]}y~(x, y) \in \eval{\eta}\} = \left(\proj(\eval{\eta}^c \cap (\R^n \times [0, 1]))\right)^c\]
    where $\R^n$ is viewed as a subspace of $\R^{n + 1}$ representing the first $n$-coordinates. By induction, $\eval{\eta}$ is open in $\R^{n + 1}$ and therefore $\eval{\eta}^c \cap (\R^n \times [0, 1])$ is closed in $\R^n \times [0, 1]$. Since $[0, 1]$ is compact, the natural projection $\proj : \R^n \times [0, 1] \to \R^n$ is a closed map, and therefore $\proj(\eval{\eta}^c \cap (\R^n \times [0, 1]))$ is a closed set in $\R^n$, consequently $\eval{\psi}$ is open. The dual case of existential quantification holds similarly, noting that $\proj$ is always an open map. And the case for logical connectives correspond to finite unions/intersections and are therefore trivial. Thus, the proof follows by induction. 
\end{proof}
The following technical lemma establishes that a similar ``radius of perturbation'' exists continuously. 

\begin{lemma}
    \label{lem: continuously choose perturbation}
    Let $\phi \in \fmlbo(\L_D)$ be a bounded formula, $D : FV(\phi) \to \IQ$ a domain of $\phi$ and $I : \term(\phi) \to \IQ$ an enclosure of $\phi$ extending $D$. There exists a continuous function $m^I_\phi : D \to [0, \sfrac{1}{4}]$ such that the following are equivalent for all $\vec{z} \in D$:
    \begin{enumerate}
        \item $\vec{z} \in \evalb{\phi}$
        \item $m^I_\phi(\vec{z}) > 0$
        \item $m^I_\phi(\vec{z}) > 0$ and $B_{D}(\vec{z}, m^I_\phi(\vec{z})) \subseteq \evalb{\phi}$
        \item $m^I_\phi(\vec{z}) > 0$ and $\vec{z} \in \evalb{\pertf[\F]{\phi}{m^I_\phi(\vec{z})}}$ for all $(I, m^I_\phi(\vec{z}))$ admissible approximations $\F$.
    \end{enumerate}
    where $B_{D}(\vec{z}, r)$ denotes $B(\vec{z}, r) \cap D$. 
\end{lemma}

\begin{proof}
    We prove the $(1)-(4)$ by simultaneous induction for all $\phi(\vec{x})$, first inducting over the number of function symbols and secondly inducting over the structural complexity, showing it for all $D, I$ at once. 
    \begin{itemize}
        \item $\phi(\vec{x})$ is function free: In this case, $m^I_\phi(\vec{x})$ can be defined as the function $\min(r_\phi(\vec{x}), \sfrac{1}{4})$ given by Lemma \ref{lem: continuously choose interior}. Properties $(1)-(3)$ are satisfied by Lemma \ref{lem: continuously choose interior} and $(4)$ is satisfied as the perturbation $\pertf[\F]{\phi}{m^I_\phi(\vec{x})}$ of a function-free formula $\phi$ is the original formula $\phi$ itself. 

         \item $\phi(\vec{x}) \equiv \psi(f(e), \vec{x})$, where $e$ is function-free and $\psi$ is quantifier-free: Consider the formula $\psi(w, \vec{x})$ equipped with the enclosure $J \coloneqq I_{w \mapsto f(e)}$ extending the domain $E \coloneqq D \cup \{w \mapsto I_{f(e)}\}$. Applying the inductive hypothesis on $\psi(w, \vec{x})$ then yields a function $m^J_{\psi} : E \to [0, \sfrac{1}{4}]$ with properties $(1) - (4)$. Define the following auxiliary continuous functions on the domain $D \subseteq \R^n$
        \begin{align*}
            m_1(\vec{z}) &= \frac{1}{4}m^J_\psi(f(e(\vec{z})), \vec{z})\\
            m_2(\vec{z}) &= \min_{\substack{\abs{\lambda - f(e(\vec{z}))} \leq 2m_1(\vec{z}) \\ \lambda \in I_{f(e)}}} m^J_\psi(\lambda, \vec{z})
        \end{align*}
        which are well defined by the extreme value theorem and noting that $f(e(\vec{z})) \in I_{f(e)}$ for $\vec{z} \in D$. Finally, $m^I_\phi(\vec{z})$ can be constructed via the following where $r_\phi(\vec{z})$ is the function given by Lemma \ref{lem: continuously choose interior} for $\phi(\vec{x})$
        \[m^I_\phi(\vec{z}) \coloneqq \min(r_\phi(\vec{z}), m_1(\vec{z}), m_2(\vec{z}))\]
        Note that $m^I_\phi$ is indeed continuous as taking pointwise minimums/minimums over compact intervals are continuous operations, so it remains to verify the equivalence of properties $(1) - (4)$. 
        \begin{itemize}
            \item $(1) \implies (2)$: Suppose $\vec{z} \in \evalb{\phi}$ holds. Since $r_\phi(\vec{z}) > 0$ follows directly by Lemma \ref{lem: continuously choose interior} and $(f(e(\vec{z})), \vec{z}) \in \eval{\psi}_E$ implies $m_1(\vec{z}) > 0$ by inductive hypothesis, it suffices to show $m_2(\vec{z}) > 0$. To this end, consider any $\lambda \in I_{f(e)}$ such that $\abs{\lambda - f(e(\vec{z}))} \leq 2m_1(\vec{z})$, we want to show $m^J_\psi(\lambda, \vec{z}) > 0$. Since $\vec{z} \in \evalb{\phi}$, this gives 
            \begin{align*}
                \vec{z} \in \evalb{\phi} &\implies (f(e(\vec{z})), \vec{z}) \in \eval{\psi}_E       \end{align*}
            Applying the inductive hypothesis $(3)$, it follows that
            \[B_{E}((f(e(\vec{z})), \vec{z}), m^J_\psi(f(e(\vec{z})), \vec{z})) \subseteq \eval{\psi}_E\]
            By construction and the fact that $m^J_\psi(f(e(\vec{z})), \vec{z}) > 0$, this gives
            \[\abs{\lambda - f(e(\vec{z}))} \leq 2m_1(\vec{z}) = \frac{1}{2}m^J_\psi(f(e(\vec{z})), \vec{z}) < m^J_\psi(f(e(\vec{z})), \vec{z})\]
            Therefore
            \[\norm{(\lambda, \vec{z}) - (f(e(\vec{z})) , \vec{z})} = \abs{\lambda - f(e(\vec{z}))} < m^J_\psi(f(e(\vec{z})), \vec{z})\]
            implying
            \[(\lambda, \vec{z}) \in B_{E}((f(e(\vec{z})), \vec{z}), m^J_\psi(f(e(\vec{z})), \vec{z})) \implies (\lambda, \vec{z}) \in \eval{\psi}_E\]
            another application of the inductive hypothesis $(2)$ then gives
            \[m^J_\psi(\lambda, \vec{z}) > 0\]
            establishing the claim. 

            \item $(2) \implies (3)$: By definition of $m^I_\phi$, $m^I_\phi(\vec{z}) > 0$ implies $r_\phi(\vec{z}) > 0$, from which $(3)$ follows by Lemma \ref{lem: continuously choose interior}.

            \item $(3) \implies (4)$: Let $\F$ be an arbitrary $(I, m^I_\phi(\vec{z}))$ admissible approximation.  Recall that the $\delta$-perturbation of $\phi(\vec{x})$ for an approximation $\F$ gives:
            
            \[\pertf[\F]{\phi}{m^I_\phi(\vec{z})} = \forall^{[-m^I_\phi(\vec{z}), m^I_\phi(\vec{z})] + \F_{f(e)}(e)} w~\pertf[\F_{w \mapsto f(e)}]{\psi}{m^I_\phi(\vec{z})}\]
            
            Let $w$ be arbitrary satisfying $\abs{w - \F_{f(e)}(e(\vec{z}))} \leq m^I_\phi(\vec{z})$, it remains to establish 
            
            \[(w, \vec{z}) \in \eval{\pertf[\F_{w \mapsto f(e)}]{\psi}{m^I_\phi(\vec{z})}}_E\]
            
            To this end, first notice that by construction of $J$, the approximation $\G \coloneqq \F_{w \mapsto f(e)}$ is $(J, m^I_{\phi}(\vec{z}))$ admissible. By admissibility of $\F$, $w$ satisfies
            \begin{align*}
                \abs{w - \F_{f(e)}(e(\vec{z}))} \leq m^I_\phi(\vec{z}) &\implies \abs{w - f(e(\vec{z}))} \leq 2m^I_\phi(\vec{z})\\
                &\implies \abs{w - f(e(\vec{z}))} \leq \frac{1}{2}m^J_\psi(f(e(\vec{z})), \vec{z})\\
                &\implies (w, \vec{z}) \in B_E((f(e(\vec{z})), \vec{z}), m^J_\psi(f(e(\vec{z})), \vec{z}))
            \end{align*}
            where $w \in I_{f(e)}$ follows from the admissibility of $\F$ and $m^I_\phi(\vec{z}) < \sfrac{1}{2}$. Since $0 < m^I_\phi(\vec{z}) \leq m^J_{\psi}(f(e(\vec{z})), \vec{z})$, this implies 
            \[B_E((f(e(\vec{z})), \vec{z}), m^J_\psi(f(e(\vec{z})), \vec{z})) \subseteq \eval{\psi}_E \implies (w, \vec{z}) \in \eval{\psi}_E\]
            By definition of $m^J_\psi(w, \vec{x})$, and applying inductive hypothesis $(4)$, this further implies
            \begin{align*}
                (w, \vec{z}) \in \eval{\pertf[\G]{\psi}{m^J_\psi(w, \vec{z})}}_E
            \end{align*}
            Thus, it remains to show $m^I_\phi(\vec{z}) \leq m^J_\psi(w, \vec{z})$. It follows by construction of $m_2(\vec{z})$ that $m_2(\vec{z}) \leq m^J_\psi(w, \vec{z})$ as $w \in I_{f(e)}$ and $\abs{w - f(e(\vec{z}))} \leq 2m_1(\vec{z})$, thus $m^I_\phi(\vec{z}) \leq m^J_\psi(w, \vec{z})$, completing the proof as desired. 

            \item $(4) \implies (1)$: This follows directly from Theorem \ref{thm: proof rules for delta perturbations} as $\pertf[\F]{\phi}{\delta} \rightarrow \phi$ is always (provably) valid for admissible approximations $\F$. 
        \end{itemize}
        
        \item $\phi(\vec{x}) \equiv \exists^{[a, b]}y\psi(y, \vec{x})$: First note that an enclosure $I$ of $\phi$ is also an enclosure of $\psi$, extending the domain $E \coloneqq D \cup \{y \mapsto I_y\}$. Similar to the previous case, $m^I_\phi$ can be constructed as follows
        \begin{align*}
            &m_2(\vec{z}) = \max_{w \in [a(\vec{z}), b(\vec{z})]} m^I_\psi(w, \vec{z})\\
            &m^I_\phi(\vec{z}) = \min(r_\phi(\vec{z}), m_2(\vec{z}))
        \end{align*}
        Since $m^I_\phi(\vec{z}) \leq r_\phi(\vec{z})$, the equivalence of properties $(1), (2), (3)$ amounts to proving $(1) \implies m_2(\vec{z}) > 0$ where $\vec{z} \in D$ satisfies $(1)$. Indeed, we have
        \[\vec{z} \in \evalb{\phi} \implies \mathbf{\exists}^{[a(\vec{z}), b(\vec{z})]} w~(w, \vec{z}) \in \eval{\psi}_E\]
        Therefore, there is some $w \in [a(\vec{z}), b(\vec{z})]$ where $m^I_\psi(w, \vec{z}) > 0$ and therefore
        \[m_2(\vec{z}) \geq m_\psi(w, \vec{z}) > 0\]
        so properties $(1), (2), (3)$ are equivalent and it suffices to show $(2) \implies (4)$. As such, assume that $(2)$ is satisfied and fix some $\vec{z} \in D$ with $m^I_\phi(\vec{z}) > 0$. Further denote $w \in [a(\vec{z}), b(\vec{z})]$ as some element attaining $m^I_\psi(w, \vec{z}) = m_2(\vec{z}) \geq m^I_\phi(\vec{z}) > 0$. Standard manipulations conclude:
        \begin{align*}
            m^I_\psi(w, \vec{z}) > 0 &\implies (w, \vec{z}) \in \eval{\pertf[\F]{\psi}{m^I_\psi(w, \vec{z})}}_E\\
            &\implies \vec{z} \in \evalb{\exists^{[a, b]}y~\pertf[\F]{\psi}{m^I_\psi(w, \vec{z})}}\\
            &\implies \vec{z} \in \evalb{\pertf[\F]{\phi}{m^I_\phi(\vec{z})}}
        \end{align*}
        
        \item $\phi(\vec{x}) \equiv \forall^{[a, b]}y~\psi(y, \vec{x})$: Similar to the case above, define:
        \begin{align*}
            &m_2(\vec{z}) = \min_{w \in [a(\vec{z}), b(\vec{z})]} m_\psi(w, \vec{z})\\
            &m^I_\phi(\vec{z}) = \min(r_\phi(\vec{z}), m_2(\vec{z}))
        \end{align*}
        and the rest of the argument is identical to the existential case above. This completes the final case in our induction and the proof is therefore complete. \qedhere
    \end{itemize}
\end{proof}
Theorem \ref{thm: atomic formulas are robust} now follows naturally.
\begin{proof}[Proof of Theorem \ref{thm: atomic formulas are robust}]
    First consider the case where $\phi \in \fml^>_B(\L_D)$, let $D$ be any domain of $\phi$ and $I$ any enclosure extending $D$, which necessarily (computably) exists by Theorem \ref{thm: compute enclosures} and $(\F_\delta)_\delta$ an arbitrary sequence of $I$ admissible approximations. By Theorem \ref{thm: proof rules for delta perturbations}, it suffices to establish that 
    \[\evalb{\phi} \subseteq \bigcup_{\delta > 0} \evalb{\pertf[\F_\delta]{\phi}{\delta}}\]
    Let $m^I_\phi : D \to [0, \sfrac{1}{4}]$ be the function given by Lemma \ref{lem: continuously choose perturbation}, and suppose $\vec{z} \in \evalb{\phi}$. It then follows by monotonicity of perturbations that for all sufficiently small $\delta > 0$ satisfying $\delta \leq m^I_\phi(\vec{z})$ the following holds
    \[\vec{z} \in \evalb{\pertf[\F_\delta]{\phi}{\delta}}\]
    since $\vec{z} \in D$ is arbitrary, this gives
    \[\evalb{\phi} = \bigcup_{\vec{z} \in \evalb{\phi}} \{\vec{z}\} \subseteq \bigcup_{\delta > 0 }\evalb{\pertf[\F_\delta]{\phi}{\delta}}\]
    thus $\phi$ is indeed $\forall$-robust. Now consider the case where $\phi(\vec{x}) \in \fml^\geq_B(L_D)$, which implies $\neg\phi(\vec{x}) \in \fml^>_B(L_D)$ when the negation is pushed through the quantifiers. It has already been established that $\neg\phi$ is $\forall$-robust, therefore it follows by definition that $\phi$ is $\exists$-robust, completing the proof.  
\end{proof}

Theorem \ref{thm: atomic formulas are robust} establishes that the collection of pure formulas is robust. On the other hand, Example \ref{ex: non-robust formula} provides the existence of non-robust formulas. Since robustness is a desired property, it is natural to ask if there exists some alternate approximation algorithm such that \emph{all} bounded formulas are robust. That is, as the error of approximation tends to $0$, the truth/validity of the $\folr$ approximants converges to that of the original formula. Fundamentally, the existence of such approximation algorithms would imply that the bounded theory ($\theoryb(\R_C)/\theoryb(\R_D)$) can be computed from $\jump{1}$ (i.e. the halting problem), and therefore resides within the lower levels ($\Delta^0_2$) of the arithmetic hierarchy \cite{Soare_2016}. In light of this observation, we show that there does not exist an approximation algorithm that achieves robustness for the theory $\theoryb(\R_C)$ of Type-Two computable functions by proving $\jump{\omega} \leq_m \theoryb(\R_C)$. That is, the bounded theory of $\R_C$ is at least as complicated as arithmetic.

\begin{definition}
    $\jump{n}$ denotes the $n$-th iterated Turing jump, and $\jump{\omega}$ denotes the uniform join of all finite Turing jumps defined via
    \[\jump{\omega} = \{\ddiamond{n, k}{}~\vert~n \in \jump{k}\}\]
    where $\ddiamond{n, k}{}$ is the standard G\"odel encoding of the pair $(n, k)$ as a single natural number. 
\end{definition}

Intuitively, $\jump{1}$ denotes the standard halting problem, $\jump{2}$ denotes the halting problem relative to an oracle for $\jump{1}$, and $\jump{n}$ denotes the $n$-th iterated halting problem. $\jump{\omega}$ then denotes the ``union'' of all finite iterations of the halting problem, i.e. all problems solvable using iterated halting problems as oracles. 

\begin{definition}
    For sets $A, B \subseteq \N$, we say $A \leq_m B$ ($A$ is $m$-reducible to $B$) if there is a computable function $f : \N \to \N$ such that for all $n \in \N$:
    \[n \in A \iff f(n) \in B\]
\end{definition}

\begin{theorem}
    \label{thm: arithmetic lower bound}
    $\theory_B(\R_C)$ is at least as difficult as arithmetic.
    \[\jump{\omega} \leq_m \theory_B(\R_C)\]
\end{theorem}

Intuitively, Theorem \ref{thm: arithmetic lower bound} holds due to the versatility of computable functions, allowing one to encode arbitrary computable predicates in $\theoryb(\R_C)$. The following lemma proves the existence of a computable space filling curve which will be useful in establishing Theorem \ref{thm: arithmetic lower bound}.

\begin{lemma}[Computable space-filling curves {\cite[Section~3.3]{DBLP:journals/apal/BagavievBBBDKKMN25}}]
    \label{lem: space-filling curve}
    For all positive $n \in \N$, there exists a computable, continuous surjection $\rho : [0, 1] \to [0, 1]^n$. Furthermore, $\rho_n$ is uniformly computable in $n$.
\end{lemma}

\begin{lemma}
    \label{lem: encode predicates as formulas}
    Let $\alpha : \N \to [0, 1)$ be the function defined by $\alpha(n) = 1 - 2^{-n}$. There exists a computable function that takes in computable predicates $R \subseteq \N^n$ and outputs corresponding formulas $\phi_R(\vec{x}) \in \fml_B(\L_C)$ of the same arity such that for all $\vec{m} \in \N^n$, we have:
    \[R(\vec{m}) \iff \R_C \models \phi_R(\alpha(\vec{m}))\]
    Where $\alpha(\vec{m}) = (\alpha(\vec{m}_1), \cdots, \alpha(\vec{m}_n))$, and $R(\vec{m})$ denotes $\vec{m} \in R$. 
\end{lemma}

\begin{proof}
    See Appendix \ref{app: proofs}.
\end{proof}

Theorem \ref{thm: arithmetic lower bound} can now be proven by interpreting $\langle \N, 0, 1, +, \cdot \rangle$ in $\R_C$ via bounded formulas.

\begin{proof}[Proof of Theorem \ref{thm: arithmetic lower bound}]
    It suffices to (effectively) interpret $\langle \N, 0, 1, +, \cdot \rangle$. That is, it is sufficient to establish the following:
    \begin{enumerate}
        \item There exists some injection $\alpha : \N \to [0, 1]$ such that its image $\alpha[\N]$ and $\alpha(0), \alpha(1)$ are all definable in $\R_C$ with bounded formulas.
        \item A formula $\text{Add}(x, y, z) \in \fmlb(\R_C)$ such that for all $n_1, n_2, n_3 \in \N$
        \[n_1 + n_2 = n_3 \iff (\alpha(n_1), \alpha(n_2), \alpha(n_3)) \in \eval{\text{Add}}\]

        \item A formula $\text{Mult}(x, y, z) \in \fmlb(\R_C)$ such that for all $n_1, n_2, n_3 \in \N$
        \[n_1 \cdot n_2 = n_3 \iff (\alpha(n_1), \alpha(n_2), \alpha(n_3)) \in \eval{\text{Mult}}\]
    \end{enumerate}
    The existence of such constructs would imply $\jump{\omega} \equiv_m \theory(\N) \leq_m \theoryb(\R_C)$ as desired. Note that by defining $\alpha : \N \to [0, 1]$ as $\alpha(n) = 1 - 2^{-n}$, conditions (2), (3) are automatically satisfied by Lemma \ref{lem: encode predicates as formulas}, and $\alpha(0) = 0, \alpha(1) = \frac{1}{2}$ are both definable constants in $\R_C$. Thus, it remains to show that the image $\alpha[\N] \subseteq [0, 1]$ is definable using bounded formulas. First construct a computable function $g : [0, 1] \to \R$ as follows:
    \begin{itemize}
        \item For all $n \in \N$, define $g(\alpha(n)) = 0$.
        \item For all $n \in \N$, set $g\left(\frac{\alpha(n) + \alpha(n + 1)}{2}\right) = 2^{-n}$. Then linearly interpolate over the intervals $\left[\alpha(n), \frac{\alpha(n) + \alpha(n + 1)}{2}\right]$ and $\left[\frac{\alpha(n) + \alpha(n + 1)}{2}, \alpha(n + 1)\right]$. As linear interpolation is computable, this defines a computable function $g: [0, 1) \to [0, 1]$.
        \item Define $g(1) = 0$. Note that for all $n$, for all $x \in [\alpha(n), 1)$, $g(x) \leq 2^{-n}$ by construction. In particular, this yields a computable modulus of continuity at $x = 1$, and therefore $g : [0, 1] \to [0, 1]$ is a computable function \cite{Weihrauch_2000}.
    \end{itemize}
    Since $\L_C$ only permits total computable functions, take $h: \R \to \R$ to be some computable extension of $g$ \cite[ Expansion~theorem]{Pour-El_Richards_2017}. $\alpha[\N]$ can now be defined as follows to complete the proof
    \[\forall^{[0, 1]}x~x \in \alpha[\N] \iff (x \neq 1 \land h(x) = 0)\qedhere\]
\end{proof}

\section{Conclusion}
This article establishes a complete approximate axiomatization for $\theoryb(\R_D)$, the bounded theory of $\R$ expanded with differentially-defined functions. By utilizing perturbation of terms and crucially leveraging the deductive power of $\dL$ \cite{DBLP:journals/jacm/PlatzerT20,platzer2024axiomatizationcompactinitialvalue}, \emph{all} numerical properties needed for such perturbations can be proven symbolically and uniformly from a finite axiomatization of $\dL$, resulting in a sound axiomatization providing complete reasoning principles up to numerical perturbations. Furthermore, robustness and convergence of such numerical approximations are guaranteed under mild syntactic restrictions for pure formulas. Lastly, the article shows that such robustness properties cannot be attained for all of $\theoryb(\R_C)$ by proving a computability-theoretic lower bound on the strength of the theory. 

For future works, it would be interesting to implement the approximation procedures presented. We believe this paper lays the foundations for the development of practical, formally verified approximations that enable symbolic reasoning beyond deciding the truth of closed sentences.  

 \textit{Acknowledgment.}
Funding has been provided by an Alexander von Humboldt Professorship and the National Science Foundation under Grant No. CCF 2220311.

\newpage
\section*{Appendix}
\appendix
\section{$\dL$ Axiomatization}
\label{app: dL axiomatization}
\begin{theorem}[\cite{DBLP:conf/lics/Platzer12b, DBLP:journals/jacm/PlatzerT20, DBLP:journals/fac/TanP21}]
    \label{thm: base axiomatization of dL}
    The following are sound axioms of $\dL$. In axioms \irref{cont}, \irref{dadj}, \irref{bdg}, the variables $y$ is fresh. In axiom \irref{bdg}, $Q(x)$ is required to be a formula of real arithmetic. 

    \begin{calculus}
        \cinferenceRule[qear|\usebox{\Rval}]{quantifier elimination real arithmetic}
        {\linferenceRule[sequent]
          {}
          {\lsequent[g]{\Gamma}{\Delta}}
        }{$\text{if}~\landfold_{\ausfml\in\Gamma} \ausfml \limply \lorfold_{\busfml\in\Delta} \busfml ~\text{is valid in \LOS[\reals]}$}%
    
        \cinferenceRule[diamond|$\didia{\cdot}$]{diamond axiom}
        {\linferenceRule[equiv]
          {\lnot\dbox{\ausprg}{\lnot \ausfml}}
          {\ddiamond{\ausprg}{\ausfml}}
        }
        {}

        \cinferenceRule[evolved|$\didia{'}$]{evolve}
        {\linferenceRule[equiv]
          {\lexists{t{\geq}0}{\ddiamond{\pupdate{\pumod{x}{y(t)}}}{p(x)}}\hspace{1cm}}
          {\ddiamond{\pevolve{\D{x}=\genDE{x}}}{p(x)}}
        }{$\m{\D{y}(t)=\genDE{y}}$}%

        \cinferenceRule[B|B$'$]{}
        {\linferenceRule[equiv]
          {\lexists{y}{\ddiamond{\pevolvein{x'=f(x)}{Q(x)}}{\rfvar(x,y)}}}
          {\ddiamond{\pevolvein{x'=f(x)}{Q(x)}}{\exists{y}\rfvar(x,y)}}
        }{\text{$y \not\in x$}}
        
        \cinferenceRule[K|K]{K axiom / modal modus ponens} %
        {\linferenceRule[impl]
          {\dbox{\alpha}{(\fvarA \limply \fvarB)}}
          {(\dbox{\alpha}{\fvarA}\limply\dbox{\alpha}{\fvarB})}
        }{}
        \cinferenceRule[V|V]{vacuous $\dbox{}{}$}
         {\linferenceRule[impl]
           {\fvarA}
           {\dbox{\alpha}{\fvarA}}
         }{\text{no free variable of $\fvarA$ is bound by $\alpha$}}
        \cinferenceRule[G|G]{$\dbox{}{}$ generalization} %
        {\linferenceRule[formula]
          {\lsequent{}{\fvarA}}
          {\lsequent{\Gamma}{\dbox{\alpha}{\fvarA}}}
        }{}
        
        \cinferenceRule[dW|dW]{}
        {\linferenceRule
          {\lsequent{\ivr}{P}}
          {\lsequent{\Gamma}{\dbox{\pevolvein{\D{x}=\genDE{x}}{\ivr}}{P}}}
        }{}
        
        \cinferenceRule[dC|dC]{differential cut}%
        {\linferenceRule[sequent]
          {\lsequent[L]{}{\dbox{\pevolvein{\D{x}=\genDE{x}}{\ivr}}{\cusfml}}
          &\lsequent[L]{}{\dbox{\pevolvein{\D{x}=\genDE{x}}{(\ivr\land \cusfml)}}{\ousfml[x]}}}
          {\lsequent[L]{}{\dbox{\pevolvein{\D{x}=\genDE{x}}{\ivr}}{\ousfml[x]}}}
        }{}
        
        \cinferenceRule[DG|DG]{differential ghost variables}
        {\linferenceRule[viuqe]
          {\dbox{\pevolvein{x'=\genDE{x}}{\ivr(x)}}{\ousfml[x](x)}}
          {\lexists{y}{\dbox{\pevolvein{x'=\genDE{x}\syssep y'=a(x)\cdot y+b(x)}{\ivr(x)}}{\ousfml[x](x)}}}
        }
        {}
                \cinferenceRule[DGi|DGi]{differential ghost variables}
        {\linferenceRule[impl]
          {\dbox{\pevolvein{x'=\genDE{x}}{\ivr(x)}}{\ousfml[x](x)}}
          {\forall{y}{\dbox{\pevolvein{x'=\genDE{x}\syssep y'= g(x, y)}{\ivr(x)}}{\ousfml[x](x)}}}
        }
        {}

        \cinferenceRule[dx|DX]{differential skip}
        {
            \linferenceRule[equiv]
            {\left(Q \rightarrow P \land \dbox{x' = f(x) \& Q}{P}\right)}
            {\dbox{x'= f(x) \& Q}{P}}
        }
        {$x' \notin P, Q$}
        
        \cinferenceRule[uniq|Uniq]{vanilla uniqueness}
        {\linferenceRule[equiv]
        {\left(\ddiamond{x'= f(x) \& Q_1}{P}\right) \land \left(\ddiamond{x'= f(x) \& Q_2}{P}\right)}
        {\ddiamond{x'= f(x) \& Q_1 \land Q_2}{P}}
        }{}

        \cinferenceRule[cont|Cont]{continuity of ODE}
        {\linferenceRule[impl]
        {x = y}
        {\left(\ddiamond{x'= f(x) \& e > 0}{x \neq y} \leftrightarrow e > 0\right)}
        }{$f(x) \neq 0$}

        \cinferenceRule[dadj|Dadj]{reverse flow of ODEs}
        {\linferenceRule[equiv]
        {\ddiamond{y' = -f(y) \& Q(y)}{y = x}}
        {\ddiamond{x' = f(x) \& Q(x)}{x = y}}
        }{}

        \cinferenceRule[ri|RI]{real induction axiom}
        {\linferenceRule[equiv]
        {\forall y\dbox{x' = f(x) \& P \lor x = y}{\left(x = y \rightarrow P \land \ddiamond{x' = f(x) \& P \lor x = y}{x \neq y}\right)}}
        {\dbox{x' = f(x)}{P}}
        }{}
        
        \cinferenceRule[bdg|BDG]{bounded differential ghost}
       {
       \linferenceRule[impll]
       {\dbox{x' = f(x), y' = g(x, y) \& Q(x)}{\norm{y}^2 \leq p(x)}}
        {\left(\dbox{x' = f(x) \& Q(x)}{P(x)} \leftrightarrow \dbox{x' = f(x), y' = g(x, y) \& Q(x)}{P(x)}\right)}
       }
       {}
    \end{calculus}
\end{theorem}
\newpage
Theorem \ref{thm: base axiomatization of dL} provides a complete record of $\dL$'s axiomatization that is utilized in earlier works \cite{DBLP:journals/jacm/PlatzerT20,DBLP:journals/fac/TanP21,platzer2024axiomatizationcompactinitialvalue} to establish the completeness properties (Theorems \ref{thm: stone weierstrass for IVPs} and \ref{thm: stone weierstrass for functions}) used in this article. The usual FOL proof rules are also listed below for completeness \cite{DBLP:journals/jar/Platzer08}.\\
\begin{calculuscollection}
\begin{calculus}
    \cinferenceRule[notl|$\lnot$\leftrule]{$\lnot$ left}
    {\linferenceRule[sequent]
      {\lsequent[L]{}{\asfml}}
      {\lsequent[L]{\lnot \asfml}{}}
    }{}%
    \cinferenceRule[andl|$\land$\leftrule]{$\land$ left}
    {\linferenceRule[sequent]
      {\lsequent[L]{\asfml , \bsfml}{}}
      {\lsequent[L]{\asfml \land \bsfml}{}}
    }{}%
    \cinferenceRule[orl|$\lor$\leftrule]{$\lor$ left}
    {\linferenceRule[sequent]
      {\lsequent[L]{\asfml}{}
        & \lsequent[L]{\bsfml}{}}
      {\lsequent[L]{\asfml \lor \bsfml}{}}
    }{}%
    \cinferenceRule[notr|$\lnot$\rightrule]{$\lnot$ right}
    {\linferenceRule[sequent]
      {\lsequent[L]{\asfml}{}}
      {\lsequent[L]{}{\lnot \asfml}}
    }{}%
    \cinferenceRule[andr|$\land$\rightrule]{$\land$ right}
    {\linferenceRule[sequent]
      {\lsequent[L]{}{\asfml}
        & \lsequent[L]{}{\bsfml}}
      {\lsequent[L]{}{\asfml \land \bsfml}}
    }{}%
    \cinferenceRule[cut|cut]{cut}
    {\linferenceRule[sequent]
      {\lsequent[L]{}{\cusfml}
      &\lsequent[L]{\cusfml}{}}
      {\lsequent[L]{}{}}
    }{}%
    \cinferenceRule[orr|$\lor$\rightrule]{$\lor$ right}
    {\linferenceRule[sequent]
      {\lsequent[L]{}{\asfml, \bsfml}}
      {\lsequent[L]{}{\asfml \lor \bsfml}}
    }{}%
\end{calculus}
\qquad
\begin{calculus}
    \cinferenceRule[implyl|$\limply$\leftrule]{$\limply$ left}
    {\linferenceRule[sequent]
      {\lsequent[L]{}{\asfml}
        & \lsequent[L]{\bsfml}{}}
      {\lsequent[L]{\asfml \limply \bsfml}{}}
    }{}%
    \cinferenceRule[alll|$\forall$\leftrule]{$\lforall{}{}$ left instantiation}
    {\linferenceRule[sequent]
      {\lsequent[L]{p(\astrm)}{}}
      {\lsequent[L]{\lforall{x}{p(x)}}{}}
        \qquad{}\qquad{}
    }{arbitrary term $\astrm$}%
    \cinferenceRule[existsl|$\exists$\leftrule]{$\lexists{}{}$ left}
    {\linferenceRule[sequent]
      {\lsequent[L]{p(y)} {}}
      {\lsequent[L]{\lexists{x}{p(x)}} {}}
    }{\m{y\not\in\Gamma{,}\Delta{,}\lexists{x}{p(x)}}}%
    \cinferenceRule[implyr|$\limply$\rightrule]{$\limply$ right}
    {\linferenceRule[sequent]
      {\lsequent[L]{\asfml}{\bsfml}}
      {\lsequent[L]{}{\asfml \limply \bsfml}}
    }{}%
    \cinferenceRule[allr|$\forall$\rightrule]{$\lforall{}{}$ right}
    {\linferenceRule[sequent]
      {\lsequent[L]{}{p(y)}}
      {\lsequent[L]{}{\lforall{x}{p(x)}}}
    }{\m{y\not\in\Gamma{,}\Delta{,}\lforall{x}{p(x)}}}%
    \cinferenceRule[existsr|$\exists$\rightrule]{$\lexists{}{}$ right}
    {\linferenceRule[sequent]
      {\lsequent[L]{}{p(\astrm)}}
      {\lsequent[L]{}{\lexists{x}{p(x)}}}
    }{arbitrary term $\astrm$}%

    \cinferenceRule[id|id]{identity}
    {\linferenceRule[sequent]
      {\lclose}
      {\lsequent[L]{\asfml}{\asfml}}
    }{}%
\end{calculus}
\end{calculuscollection}

\section{Proofs}
\label{app: proofs}
This section provides all omitted proofs. 

\begin{proof}[Proof of Theorem \ref{thm: stone weierstrass for functions}]
    Recall that the differentially-defined function $h$ is annotated with some IVP having provable global existence. Without loss of generality, further assume that the IVP starts at time $t_0 = 0$, i.e. it is of the form
    \begin{align*}
        &x' = f(x)\\
        &x(0) = X
    \end{align*}
    Since the IVP has global existence, its solution $\phi : \R \to \R^n$ is well-defined and Theorem \ref{thm: stone weierstrass for IVPs} applies. Consequently, we may compute polynomials $\theta^{+}(t), \theta^{-}(t) \in \Q^n[t]$ such that the following are provable
    \begin{align}
        x = X \land t = 0 &\rightarrow \dbox{x' = f(x), t' = 1 \& t \leq T}{\norm{x - \theta^{+}(t)}^2 < (\sfrac{\eps}{4})^2} \label{fml: fwd time approx}\\
        x = X \land t = 0 &\rightarrow \dbox{x' = -f(x), t' = -1 \& t \geq -T}{\norm{x - \theta^{-}(t)}^2 < (\sfrac{\eps}{4})^2} \label{fml: bwd time approx}
    \end{align}
    Let $p \in \Q[t]$ be any polynomial such that $\norm{p - h}_{[-T, T]} < \sfrac{\eps}{4}$ holds, which necessarily (computably) exists by the classical Stone-Weierstrass theorem. It then follows that the following are valid formulas of $\folr$ and therefore provable with axiom $\irref{qear}$. 
    \begin{align}
         0 \leq t \leq T &\rightarrow \abs{p(t) - \theta^+_0(t)} < \sfrac{\eps}{2} \label{fml: fwd time poly}\\
        -T \leq t \leq 0 &\rightarrow \abs{p(t) - \theta^-_0(t)} < \sfrac{\eps}{2} \label{fml: bwd time poly}
    \end{align}
    where $\theta^+_0, \theta^-_0$ denote the first-coordinate projections of $\theta^+$ and $\theta^-$ respectively. On the other hand, applying axiom \irref{FI} in conjunction with the formulas (\ref{fml: fwd time approx}), (\ref{fml: bwd time approx}) proves the following 
    \begin{align}
        0 \leq t \leq T &\rightarrow \abs{h(t) - \theta^+_0(t)} < \sfrac{\eps}{4} \label{fml: fwd time error bound}\\
        -T \leq t \leq 0 &\rightarrow \abs{h(t) - \theta^-_0(t)} < \sfrac{\eps}{4} \label{fml: bwd time error bound}
    \end{align}
    combining formulas (\ref{fml: fwd time error bound})+(\ref{fml: fwd time poly}) and (\ref{fml: bwd time error bound})+(\ref{fml: bwd time poly}) proves the following
    \begin{align*}
        0 \leq t \leq T &\rightarrow \abs{h(t) - p(t)} < \sfrac{3\eps}{4}\\
        -T \leq t \leq 0 &\rightarrow \abs{h(t) - p(t)} < \sfrac{3\eps}{4}
    \end{align*}
    which then combines to prove the desired formula (\ref{fml: error bound}). 
\end{proof}

\begin{proof}[Proof of Lemma \ref{lem: mont perturbations}]
    Let $\psi(\vec{x}, \vec{y})$ denote the quantifier-free part of $\phi(\vec{x})$ (recall that $\phi(\vec{x})$ is assumed to be in prenex normal form), where $\vec{x}$ are the free variables and $\vec{y}$ are the bound variables. By construction of the perturbations, it suffices to show that $\pertf[\F]{\psi}{\delta} \rightarrow \pertf[\G]{\psi}{\frac{\delta - \norm{\F}_I}{2}}$ is valid in $I$. This can be proven inductively by first inducting on the number of function terms and then on the structural complexity:
    \begin{itemize}
        \item The base case where $\phi$ does not contain functions is trivial, since perturbation is the identity in this case. 
        \item The cases of boolean connectives and quantifiers hold trivially by inductive hypothesis.
        \item Consider the non-trivial case of $\psi \equiv t(f(e), \vec{z}) \succeq 0$ where $\vec{z} = (\vec{x}, \vec{y})$ and the term $e(\vec{z})$ is function-free. The perturbations corresponding to $\F, \G$ are:
        \begin{align*}
            &\pert^\forall_{\F}\left(t(f(e), \vec{z}) \succeq 0, \delta\right) = \forall^{[-\delta, \delta] + \F_{f(e)}(e)} w~\pert^\forall_{\F_{w \mapsto f(e)}}\left(t(w, \vec{z}) \succeq 0, \delta\right)\\
            &\pert^\forall_{\G}\left(t(f(e), \vec{z}) \succeq 0, \frac{\delta - \norm{\F}_I}{2}\right) = \forall^{\left[-\frac{\delta - \norm{\F}_I}{2}, \frac{\delta - \norm{\F}_I}{2}\right] + \G_{f(e)}(e)} w~ \pert^\forall_{\G_{w \mapsto f(e)}}\left(t(w, \vec{z}) \succeq 0, \frac{\delta - \norm{\F}_I}{2}\right)
        \end{align*}
        By applying the inductive hypothesis on $t(w, \vec{z}) \succeq 0$ with enclosure $I_{w \mapsto f(e)}$, it suffices to show that $\abs{w - \G_{f(e)}(e)} \leq \frac{\delta - \norm{\F}_I}{2} \implies \abs{w - \F_{f(e)}(e)} \leq \delta$. To this end, let $w \in \R$ be a witness to $\abs{w - \G_{f(e)}(e)} \leq \frac{\delta - \norm{\F}_I}{2}$. Since $\G$ is $(I, \frac{\delta - \norm{\F}_I}{2})$ admissible, this necessarily implies that $\abs{w - f(e)} < \delta - \norm{\F}_I$. Consequently, the triangle inequality gives
        \[\abs{w - \F_{f(e)}(e)} \leq \abs{w - f(e)} + \norm{\F}_I < \delta\]
        which completes the proof by induction.
    \end{itemize}
    To obtain corresponding results for existential perturbations, one could either follow the proof above but use $\pert^\exists$ instead, or utilize Lemma \ref{lem: perturbation of negated formulas} and directly apply duality.  
\end{proof}

\begin{proof}[Proof of Lemma \ref{lem: encode predicates as formulas}]
    Let $R \subseteq \N^n$ be a computable predicate. First construct a computable function $f: [0, 1]^n \to \R$ satisfying $f(\alpha(\vec{m})) > 0 \iff R(\vec{m})$ as follows:
    \begin{enumerate}
        \item For $\vec{m} \in \N^n$, define $f(\alpha(\vec{m})) = \frac{1}{\norm{\vec{m}}}\chi_{R}(\vec{m})$, where $\norm{\vec{m}}$ denotes the standard Euclidean norm, and $\chi_R : \N^n \to \{0, 1\}$ is the characteristic function of $R$. I.e. $\chi_R(\vec{m}) = 1 \iff \vec{m} \in R$.

        \item As $([\alpha(i), \alpha(i + 1)))_{i \in \N}$ forms a partition of $[0, 1)$, it follows that $\left(\prod_{1 \leq i \leq n} [\alpha(\vec{m}_i), \alpha(\vec{m}_i + 1))\right)_{\vec{m} \in \N^n}$ forms a partition of $[0, 1)^n$. For each box of the form $\prod_{1 \leq i \leq n} [\alpha(\vec{m}_i), \alpha(\vec{m}_i + 1))$, define $f$ on $\prod_{1 \leq i \leq n} [\alpha(\vec{m}_i), \alpha(\vec{m}_i + 1))$ by multilinear interpolation from the vertices. I.e. For $\vec{x} \in \prod_{1 \leq i \leq n} [\alpha(\vec{m}_i), \alpha(\vec{m}_i + 1))$, $f(\vec{x})$ is the weighted average of the values of $f$ at the vertices of the box $\prod_{1 \leq i \leq n} [\alpha(\vec{m}_i), \alpha(\vec{m}_i + 1))$, which are already defined, note that this interpolation procedure is computable. Doing this for all boxes defines a computable function $f: [0, 1)^n \to \R$. 

        \item For $\vec{x} \in [0, 1]^n \setminus [0, 1)^n$, define $f(\vec{x}) = 0$. Note that this preserves the continuity of $f$, since for any sequence of points $(z_k)_k \to \vec{x}$ with $(z_k)_k \subseteq [0, 1)^n$, it is necessarily the case that the vectors $\vec{m}$ in part (1) of the construction which forms the vertices of the box that $z_i$ belongs to diverges to $\infty$ in norm. From this it follows by construction that $f(z_i)_i \to 0$. 
        
        It remains to verify that $f$ is still computable. To this end, let $\vec{x} \in [0, 1]^n$ be arbitrary, $(\vec{q}_i)_{i \in \N} \subseteq \Q^n$ be a name of $\vec{x}$ and $k \in \N$ be arbitrary. To prove that $f$ is computable, it suffices to (uniformly) compute some $q \in \Q^n$ such that $\abs{q - f(\vec{x})} < 2^{-k}$. This can be achieved by first searching for some finite index $n_1 \in \N$ such that at least one of the following is true:
            \begin{itemize}
                \item $B(\vec{q}_{n_1}, 2^{-n_1}) \cap [0, 1]^n \subseteq [0, 1)^n$ I.e. $\vec{q}_{n_1}$ witnesses $\vec{x} \in [0, 1)^n$, $f(\vec{x})$ can then be computed by the computable procedure in construction (2).

                \item For all $\vec{y} \in B(\vec{q}_{n_1}, 2^{-n_1}) \cap [0, 1)^n$, it is the case that $\norm{\vec{m}} > 2^{k}$, where $\vec{m} \in \N^n$ characterizes the box that $\vec{y}$ belongs to in the sense of (2). In this case, we can simply output the desired value $q$ as $q = 0$. This is because $\vec{x} \in B(\vec{q}_{n_1}, 2^{-n_1})$, so either $\vec{x} \in [0, 1]^n \setminus [0, 1)^n$, in which case $f(\vec{x})$ is defined to be $0$. Or $\vec{x} \in B(\vec{q}_{n_1}, 2^{-n_1}) \cap [0, 1)^n$, so the value of $f$ at the vertices of $\vec{x}$ will all be smaller than $2^{-k}$ as $f(\vec{x}) \leq \frac{1}{\norm{\vec{m}}} < 2^{-k}$.
            \end{itemize}
        It therefore remains to show that such an $n_1$ exists. Notice that if $\vec{x} \in [0, 1)^n$ holds, then some such $n_1$ necessarily exists. So suppose that $\vec{x} \in [0, 1]^n \setminus [0, 1)^n$, then the sequence $\vec{q}_n$ converges to some element in $[0, 1]^n \setminus [0, 1)^n$. Suppose for the sake of contradiction that the second condition does not hold for all $n_1$, then this gives a sequence $(\vec{y}_i)_i \subseteq [0, 1)^n$ converging to $\vec{x}$ such that the norms $\norm{\vec{m}_i}$ does not diverge where $\vec{m}_i \in \N^n$ are the corners of $\vec{y}_i$'s, contradicting the construction of $\alpha(\vec{m})$. Hence, such a $n_1$ necessarily exists, and the search is therefore computable. 
    \end{enumerate} 
    The construction yields a computable function $f : [0, 1]^n \to \R$ such that $R(\vec{m}) \iff f(\alpha(\vec{m})) > 0$. However, recall that $\L_C$ only allows for univariate functions, yet $f$ is $n$-ary. To this end, let $\rho : [0, 1] \to [0, 1]^n$ denote the computable space-filling curve given by Lemma \ref{lem: space-filling curve}, and denote by $\rho_i$ the natural (computable) projections for $1 \leq i \leq n$. As $f$ is computable, it follows that the function $h \coloneqq f \circ \rho : [0, 1] \to \R$ is computable as well. Without loss of generality, we may further computably extend $h$'s domain to all of $\R$ \cite[Extension Theorem]{Pour-El_Richards_2017}. It then follows by construction that the following equivalences hold:
    \[R(\vec{m}) \iff f(\alpha(\vec{m})) > 0 \iff \exists^{[0, 1]}y \left(\bigwedge_{1 \leq i \leq n} \rho_i(y) = \alpha_i(\vec{m}) \land h(y) > 0\right)\]
    which is a bounded formula in the language $\L_C$, completing the proof. 
\end{proof}

\bibliographystyle{ACM-Reference-Format}
\bibliography{references_arxiv}

\end{document}